\documentclass[12pt,aip,jmp,graphicx]{revtex4-1}

\usepackage{bm}
\usepackage[utf8]{inputenc}
\usepackage[T1]{fontenc}
\usepackage{mathptmx}
\usepackage{etoolbox}
\usepackage{amsmath, amsfonts, amssymb, amsthm, mathtools, physics}
\usepackage{float}  

\newcommand{\br}[1]{\mathopen{}\left(#1\right)\mathclose{}}
\newcommand{\wt}[1]{ \widetilde{ #1 } }

\def\RR{\mathbb{R}}
\def\CC{\mathbb{C}}
\def\NN{\mathbb{N}}
\def\HH{\mathcal{H}}
\def\UU{\mathcal{U}}
\def\Fock{\mathcal{F}}
\def\SS{\mathcal{S}}
\def\DD{\mathcal{D}}
\def\Tau{\mathcal{T}}
\let\Om\Omega
\let\om\omega
\let\lam\lambda
\let\D\Delta
\def\bk{{\bm k}}
\def\bx{{\bm x}}

\theoremstyle{plain}
\newtheorem{theorem}{Theorem}[section] 
\newtheorem{lemma}{Lemma}[section]
\theoremstyle{remark}
\newtheorem*{remark}{Remark}
\theoremstyle{definition}
\newtheorem{condition}{Condition}[section]
\newtheorem{example}{Example}[section]

\begin{document}

\title{Weak coupling limit for quantum systems with unbounded weakly commuting system operators\footnote{\small{This article may be downloaded for personal use only. Any other use requires prior permission of the author and AIP Publishing. This article appeared in J. Math. Phys. 66, 042101 (2025) and may be found at \url{https://doi.org/10.1063/5.0239525}}}} 

\author{Ilya A. Lopatin}
\email{ilya.lopatin2606@gmail.com}
\affiliation{Department of Mathematical Methods for Quantum Technologies, Steklov Mathematical Institute of Russian Academy of Sciences}
\author{Alexander N. Pechen}
\email{apechen@gmail.com}
\affiliation{Department of Mathematical Methods for Quantum Technologies, Steklov Mathematical Institute of Russian Academy of Sciences}
\affiliation{University of Science and Technology MISIS}
\affiliation{Ivannikov Institute for System Programming of the Russian Academy of Sciences}

\begin{abstract}
This work is devoted to a rigorous analysis of the weak coupling limit (WCL) for the reduced dynamics of an open infinite-dimensional quantum system interacting with electromagnetic field  or a reservoir formed by Fermi or Bose particles in the dipole approximation. The free system Hamiltonian and the system part of the Hamiltonian describing interaction with the reservoir are considered as unbounded operators with continuous spectrum which are commuting in a weak sense. We derive in the weak coupling limit the reservoir statistics, which is determined by whose terms in the multi-point correlation functions of the reservoir which are non-zero in the WCL. Then we prove that the resulting reduced system dynamics converges to unitary dynamics (such behavior sometimes called as Quantum Cheshire Cat effect) with a modified Hamiltonian which can be interpreted as a Lamb shift to the original Hamiltonian. We obtain exact form of the modified Hamiltonian and estimate the rate of convergence to the limiting dynamics. For Fermi reservoir, we prove the convergence of the full Dyson series. For Bose case the convergence is understood term by term.
\end{abstract}

\pacs{}

\maketitle 

\section{Introduction}\label{Sec1:Intro}
In the theory of open quantum systems the weak coupling limit (WCL) is the most studied limit which allows to derive a Gorini-Kossakowski-Sudarshan-Lindblad (GKSL) master equation\cite{breuer2002theory} starting from an exact fully quantum Hamiltonian describing interaction between the quantum system and its surrounding quantum reservoir. A rigorous derivation of the master equation in the WCL goes back to the fundamental works of E.B. Davies\cite{Davies::markovian-MasEq, Davies_1976}, where the case of an $N$-level atom  coupled to an infinite reservoir of Fermi particles was considered and a GKSL master equation in the WCL was derived. The WCL relies on the Bogoliubov-van Hove limit\cite{Bogolubov::ProbDynamThStatPh, VanHove::EnCorrPersistPerturb} for a model of a system interacting with the bath. Denoting by $\lambda$ the interaction coupling strength, one considers time rescaling $t\to t/\lambda^2$ and then studies the reduced system dynamics in the limit as $\lambda\to 0$. It is equivalent to consider coupling strength as going to zero, $\lambda\to 0$, and time on which one studies the dynamics as going to infinity, $t\to\infty$, however not independently but in a related way so that the product $\lambda^2t$ remains finite and determines a new time scale for the dynamics. It turns out that when the coupling strength goes to zero, on a unit time its contribution to the reduced system dynamics is negligible. However, on the time scale $\sim 1/\lambda^2$ the contribution of the weak interaction becomes significant and determines a non-trivial reduced system dynamics.

Various works on the WCL were performed after that. A similarity between the singular coupling and the weak coupling limits was established\cite{Palmer_1977}. The weak coupling limit was used to describe irreversible thermodynamics and derive kinetic equations for open quantum systems\cite{Spohn_Lebowitz_1978,Spohn_1980}. It was shown that among several different Markov approximations for a quantum system
weakly coupled to a thermal reservoir in general only the rigorous approximation given by E.B. Davies preserves positivity\cite{Dumcke_Spohn_1979}.  Master equations for two-level system interacting with a boson reservoir in the dipole approximation in the WCL were studied~\cite{Jaksik::MasterEquations}. Conditions for complete positivity for two non-interacting systems embedded in a heat bath in the WCL were obtained\cite{Benatti_Floreanini_2001}. Mathematical theory with the WCL for the Wigner-Weisskopf atom was developed~\cite{Jaksic2006}. A generalization of the WCL to include non-secular terms was performed\cite{Trushechkin_2021}. 

The stochastic limit approach, which allows to describe not only the reduced system dynamics, but also the compound dynamics of the system and the reservoir, was developed\cite{ALVBook,Accardi_Frigerio_Lu_1990,Accardi_Frigerio_Lu_1991} and applied to various quantum systems\cite{Accardi_Kozyrev_Volovich_1997,Accardi_Kozyrev_Volovich_1999,Accardi_Kozyrev_Pechen_2006}. In particular, the polaron model and non-relativistic electromagnetics with unbounded system Hamiltonians were considered\cite{Accardi_Kozyrev_Volovich_1999}. 
The weak coupling and the stochastic limit were both described\cite{Derezinski_DeRoeck_2007}. In the stochastic limit approach one proves that the reservoir in the WCL is described by quantum white noise operators with commutation relations $[b(t),b^{\dagger} (\tau)] \sim \delta( t-\tau)$. Using this property, one can describe not only the reduced system dynamics in the WCL, but also the joint dynamics of the system and the reservoir. An approach to study higher order corrections to the WCL and to the stochastic limit was developed\cite{Pechen_Volovich_2002,Pechen_2004}. In this approach, higher order terms in $\lambda$ to the limiting dynamics are described by quantum multipole noise, i.e. operator valued distributions $b_n{^\pm}$ with commutation relations $\comm{b_n^{-} (t) }{b_m^{+} (\tau)}\approx \delta_{nm}i^n\delta^{(n)}(t-\tau)$, where $\delta_{nm}$ is Kronecker delta and $\delta^{(n)}(t-\tau)$ is $n$-th derivative of Dirac delta-function. Non-petrubative effects in the investigation of corrections to the WCL were discovered\cite{Teretenkov_2021}.

Many works on the WCL consider the case of an $N$-level quantum system coupled to an infinite heat bath. Quantum systems with an infinite-dimensional Hilbert space and with unbounded system operators in the free and interaction parts of the joint Hamiltonian were considered for example in the stochastic limit approach\cite{Accardi_Kozyrev_Volovich_1999, ALVBook}. An important rigorous analysis was performed in  Ref.~\onlinecite{Jaksik::FermiGoldenRule}, where a general Hamiltonian with interaction between the system and the field was studied and for an unbounded system free Hamiltonians and bounded system part of the interaction Hamiltonian, self-adjointness of a closure of the total Hamiltonian was proved. After that spectral deformation and Fermi's Golden Rule were derived.

In general, a rigorous analysis and derivation of the reduced system dynamics for the WCL when both system free Hamiltonian and system part of the interaction Hamiltonian are unbounded operators is less studied. However, such systems include important for applications examples. In this work, we consider the case of an infinite quantum system coupled to the reservoir of free Fermi or Bose particles with unbounded weakly commuting system operators with continuous spectrum in the free and interaction parts of the total system+reservoir Hamiltonian. A particular physical motivation comes from the model describing interaction of an atom with electromagnetic field in the dipole approximation~\cite{Arai::HamiltonianQuantumElectrodynamics}, where one considers interaction of a non-relativistic particle (atom) with an external electromagnetic field. The Hilbert space of the electromagnetic field in the Coulomb gauge is $\HH^{ \text{em} } = \Fock \otimes \Fock$, where $\Fock$ is the Fock space over $L^2\br{\RR^3}$. The Hilbert space of the total system is $\HH=L^2\br{\RR^3}\otimes \HH^\text{em}$. The free atom Hamiltonian is $H^{S}=-\Delta/2M$, where $M$ is mass of the atom. Denote by $\bm{e}^{(p)}$ for $p=1, 2$ two orthogonal polarization vectors of the field and consider the vector operator $ \mathbf{A}  = \sum_{p= 1, 2} a^{\dagger}_p  ( \mathbf{e}^{ (p) } f ) + a_p (   \mathbf{e}^{(p)} \overline{ f} ) $, where $f \in L_2 ( \RR^3 ) $ is some function and bar denotes its complex conjugate, with a particular form $ f = 1/\sqrt{(2(2\pi)^3} e^{-i \bk \bx } v / \sqrt{ \abs{\bk} } $, where  $v$ describes ultraviolet cut-off (an arbitrary rotationally invariant real function on $\RR^3$ satisfying $ \| v/\sqrt{ \bm{k} }  \| <\infty, \norm{ v/\bm{k} } < \infty$. Then the dipole approximation is described by the interaction $V = ( -i \nabla) \cdot \bm{A}$. In Ref.~\onlinecite{Arai::HamiltonianQuantumElectrodynamics}, self-adjointness of the total Hamiltonian was proved. 

Let us highlight the following features of this model. The system part of the interaction, which is $( -i \nabla) $, commutes with the free system Hamiltonian, which is $-\Delta/2M$. And both these operators are unbounded and having continuous spectra (resp, $\mathbb R^3$ and $\mathbb R_+$). In this paper, taking into account these features, we consider a generalization of this model. Namely, we consider a quantum system interacting with a Fermi or Bose reservoir such that the free system Hamiltonian and the system operators describing interaction of the system with the reservoir are unbounded and commuting in the weak sense. For this general model, we rigorously consider the weak coupling limit and study Dyson series for the reduced system dynamics. Each term of the Dyson series is determined by certain correlation functions of the reservoir. First, we derive the limit of the reservoir correlation functions appearing in the Dyson series of the reduced system dynamics in the weak coupling limit. When we prove absolute and uniform convergence of the Dyson series, summarize all the terms which are non-zero in the limit, and prove that the resulting reduced dynamics is unitary. We obtain the exact form of the limiting generator and show that is a Hermitian operator (a modified system free Hamiltonian). Thus the reduced system dynamics in this case has no dissipative part. Such behavior with purely oscillating regime, described by Leggett\cite{Leggett_1987} et al. for a dissipative two-state system, is known in the theory of open quantum systems and sometimes called as Quantum Cheshire Cat effect\cite{ALVBook, Accardi_Kozyrev_Volovich_1997}. We also estimate the convergence rate. The modified Hamiltonian has the form of the initial system free Hamiltonian plus a weakly commuting with it additional term (a.k.a. Lamb shift). Numerical simulations for a Gaussian wave-packet are performed to show differences in the behavior of the wave-packet under the action of the initial system free Hamiltonian and of the limiting modified Hamiltonian: they show that the modified Hamiltonian leads to a non-symmetric spreading of the wave-packet.

The structure of the paper is the following. In Sec.~\ref{Sec2:Model} a precise definition of the considered model is given. In Sec.~\ref{sec:DS-definition} we provide the Dyson series for model and introduce some necessary for the analysis notations. Sec.~\ref{Sec3:Statistics} contains derivation of the non-zero in the limit correlation functions of the reservoir and of the limiting reservoir statistics. Some details of the proof are given in Appendix \ref{app:ProofLemma}. In Sec.~\ref{Sec4:Reduced} the reduced limiting dynamics of the system is derived, which is shown to be a unitary dynamics with some modified system Hamiltonian. Sec.~\ref{Sec6:Simulations} contains numerical simulations of the free and the limiting dynamics of a Gaussian wave-packet for some particular cases. In Discussion Sec.~\ref{Sec5:Discussion} the possibility to consider Bose reservoir and states other than thermal equilibrium, improving some estimates, and computation in the zero temperature limit with appearance of the Lamb shift, are outlined.

\section{The Model and Some General Definitions} \label{Sec2:Model}
In this section, we describe the mathematical model of a quantum system interacting with Fermi or Bose reservoir or an electromagnetic field. Let $\HH^e = L_2 ( \RR^3 ) $ be Hilbert space associated with one particle of the reservoir with scalar product $( \cdot , \cdot)$ (linear with respect to the first argument) and let $\Fock_{\pm} = \Fock_{\pm} ( \HH^e )$ be symmetric (Bose, $+$) and anti-symmetric (Fermi, $-$) Fock space. The Hilbert space for electromagnetic field is 2-fold tensor product\cite{Arai::HamiltonianQuantumElectrodynamics, CANNON:ElectromagneticField} $\HH^R = \Fock_{+} \otimes \Fock_{+}$. By $a^{\dagger} ( \bm k )$ and $a ( \bm k)$ denote operators of creation and annihilation of particle with momentum $\bm k$ respectively. For $f \in \HH^e$, define $a^{\dagger} ( f ) = \int_{\RR^3} f ( k ) a^{\dagger} (k) dk $ and $a (f) = \int_{\RR^3} \overline{ f (k) } a(k)dk $ . Polarization of massless particles (e.g. photons) can be taken into account by using $ a^{\#}_{1} (f) = a^{\#} (f) \otimes \mathbb{I}$ and $a^{\#}_{2} (f) = \mathbb{I} \otimes  a^{\#} (f)$. Case without polarisation corresponds to $\HH^R = \Fock_{\pm}$.  As the one particle Hamiltonian $h$, we consider two cases. The first case is the Hamiltonian of a massless particle $h_1$, which has the form of multiplication by $h_1 = \abs{ \bm k}$ in the momentum representation. The second case is the Hamiltonian of a free particle with mass, which has the form of multiplication by $h_2 = \abs{ \bm k}^2$ in the momentum representation.  The operators $h_1$ and $ h_2$ are unbounded and self-adjoint in their maximal domains of definition. Let $H^R = \dd \Gamma (h)$ be second quantization of $h = h_1, h_2$.

Let $\HH^{S}$ be Hilbert space (complex separable and generally infinite dimensional) associated with a quantum system. By $H^S: \HH^{S} \supseteq D (H^S) \to \HH^S $ denote free Hamiltonian of the system, which in general is an unbounded operator.

By $\SS  ( \RR^3 ) $ denote the Schwartz space of functions over $\RR^3$ and let  $C_0 (\RR^3)$ be the space of complex-valued functions in $\mathbb R^3$ with compact support. Let $\mathfrak{G}$ be (dense in $\HH^e$) space of functions which belong to $C_0 (\RR^3)$ in the momentum representation and belong to $\SS (\RR^3)$ in the coordinate representation. 

Denote by $S \br{k}$ the set of all permutation of $k$ elements. By $\D_k (\tau)$ and $S_k ( \tau)$ denote the following sets in $\RR^k$:
\begin{eqnarray*}
    \Delta_k \br{\tau} &=& \left\{ t_1, t_2, \dots, t_k: \tau > t_1 > t_2 > \dots > t_k > 0  \right\}~, \\
    S_k \br{\tau} &=& \left\{ t_1, \dots, t_k > 0: t_1 + t_2 + \dots + t_k < \tau  \right\}~.
\end{eqnarray*} 

Consider interaction between the system and the reservoir as having the following form:
\begin{equation}\label{eq:Interaction}
V = \sum_{j=1}^{\nu} Q_j \otimes F_j ~,
\end{equation}
where $Q_j: \HH^S \supseteq D_0  \to \HH^S$ are some operators in the system Hilbert space which are in general unbounded and  $D_0$ dense in $\HH^S$, $\nu < \infty$. For the reservoir part we consider 
\[
F_j = \sum_{p=1}^2 a^{\dagger}_p ( f_{jp} ) + a_p ( g_{jp} )~.
\]
 
We assume that the interaction operator \eqref{eq:Interaction} satisfies the following conditions.
 
\begin{condition} \label{cond:PermutForHermite}
There exists $\sigma \in S \br{\nu}$, $\sigma^2 = \text{id}$ such that 
\begin{eqnarray*}
    Q_j^{\dagger} &=& Q_{\sigma (j) }~, \\
    F_j^{\dagger} &=& F_{\sigma (j) }~.
\end{eqnarray*} 
\end{condition}

\begin{condition} \label{cond:Formfactors}
 The form-factors $f_{jp}, g_{jp}$  belong to $\mathfrak G$ for all $j = 1, 2, \dots, \nu$, $p=1, 2$.
\end{condition}

\begin{condition} \label{cond:Interaction} 
There exists a domain $\DD \subseteq D_0 \cap D (H^S)$ dense in  $\HH^S$ such that: \\
1) $e^{ \pm i t H^S}, Q_j: \DD \to \DD $, $j=1, \dots, \nu$  for all $t \in \RR$ ; \\
2) for any $f \in \DD$ there exists a constant $c_f$ such that $\| Q_{j_1} Q_{j_2} \dots Q_{j_k} f \| \leq c_f^k \| f \|$ for $k = 1, 2, \dots$; \\
3) for any $t \in \RR$ and for any $f \in \DD$ holds $ e^{i H^S t} Q_j e^{-i H^S t} f =  Q_j f$.
\end{condition}

\begin{remark}
Condition \ref{cond:PermutForHermite}  is used to make the interaction operator \eqref{eq:Interaction} be a Hermitian (or symmetric) operator. Condition \ref{cond:Formfactors} is a technical condition necessary for performing some calculations. Condition \ref{cond:Interaction}  is a weak form of commutativity for unbounded operators that corresponds to the Generalized Rotating Wave Approximation (GRWA) with zero Bohr frequency  considered for discrete spectrum Hamiltonians in Sec.~4.10 of Ref.~\cite{ALVBook}. As it is  known, the condition of weak commutativity of two self-adjoint operators does not imply  commutativity of the unitary groups generated by this operators (as e.g., was shown in the famous Nelson's example\cite{Nelson::ExampleUnboundedOpCommutative, SimonReed::MethodModernMathPhI}). Any system with bounded operators $H^S$ and $Q_j$ such that $\comm{H^S}{Q_j} =0 $  for all $j$ where $\comm{\cdot}{\cdot}$ denotes commutator, also satisfies Condition~\ref{cond:Interaction}. 
\end{remark}

\begin{example} \label{example:CommutativeCase}
An example which satisfies the Condition~\ref{cond:Interaction} with an unbounded Hamiltonian is with the system free Hamiltonian $H^S = -\Delta/ \br{2M}$ and $Q_j = -i \partial_j$, $j = 1, 2, 3$. These operators describe interaction of an atom with electromagnetic field in the dipole approximation \cite{Arai::HamiltonianQuantumElectrodynamics}. For this case we can set $\DD := \mathfrak{G}$ and for any $f\in\DD$ considered in the momentum representation constant $c_f$ from Condition \ref{cond:Interaction} can be chosen as radius $R$ of the ball $B_{R}(0)$ in the system space centered at the origin such that $B_{R} (0) \supset  \mathop{\rm supp} f $.
\end{example}

For the algebra of system observables consider
\[
\UU^S = \left\{ \sum_{j=1}^{k} \br{\cdot, \phi_j} \psi_j, k < \infty, \phi_j, \psi_j,  \in \DD \right\}.
\]

\begin{remark}
Since the domain $\DD$ is dense in $\HH^S$, the system observable algebra $\UU^S$ is dense in the algebra of continuous operators in strong operator topology. 
\end{remark}

By $\UU^R$ we denote CAR  (canonical anticommutation relations) algebra of the reservoir for a Fermi reservoir or CCR (canonical commutation relations) algebra for a Bose reservoir. States of the reservoir are defined as linear functionals $\om: \UU^R \to \CC$ which are continuous, positive (i.e., $\om ( A^{\dagger } A ) \geq 0$) and normalized (i.e., $\om ( \mathbb I ) = 1$).  We consider states $\om$ of the reservoir which satisfy the following conditions.

\begin{condition} \label{cond:InitState} We consider states of the reservoir $\om$ which are:
\begin{itemize}
\item \textit{Quasi free}. That is, for any even number $2k$ one has
\begin{equation} \label{eq:quasi-free}
\om \big( a^{\#}_{p_1} (f_1) a^{\#}_{p_2} (f_2) \dots a^{\#}_{p_{2k}} (f_{2k}) \big)  =  \sum_{\sigma} \mathop{sgn} (\sigma) \prod_{j=1}^k \om \big( a^{\#}_{p_{\sigma (2j-1) }} ( f_{\sigma (2j-1)} ) a^{\#}_{p_{\sigma (2j)}} ( f_{\sigma (2j)} )  \big) ~,
\end{equation}
where the sum is taken over all partitioning of $\{1, 2, \dots, 2k\}$ into pairs $\br{n_j, m_j}$ with $n_j < m_j$. Permutation $\sigma$ is defined as $\sigma \br{2j-1} = n_j$, $\sigma \br{2j} = m_j$. And for any odd number $2k-1$,
\[
\om \big(  a^{\#}_{p_1} (f_1) a^{\#}_{p_2} (f_2) \dots a^{\#}_{p_{2k-1}} ( f_{2k-1} ) \big) = 0 ~.
\]
\item \textit{Gauge invariant}. That is,
\begin{eqnarray*}
\om \big( a_{p_1} (f) a_{p_2} (g) \big) &=& \om \big( a^{\dagger}_{p_1} (f) a^{\dagger}_{p_2} (g)  \big)  = 0 ~,\\
\om \big( a^{\dagger}_{p_1} (f) a_{p_2} (g) \big) &=& \delta_{p_1 p_2} \br{f, \mathbf{B} g} ~,
\end{eqnarray*}
where $\mathbf{B}$ is some operator in $\HH^e$ such that $0 \leq \mathbf{B} \leq \mathbb I$. We consider  $\mathbf{B}$ as a function of the operator $h$, $\mathbf{B} =  \rho \br{h}$, where $\rho \in \SS \br{\RR}$. In the momentum representation we have
\[
    \om  \big( a^{\dagger}_{p_1} (f) a_{p_2} (g) \big)  = \delta_{p_1 p_2} \int\limits_{\RR^3} \rho (k) f (k) \overline{g (k)}  dk~.
\]
\end{itemize}
\end{condition}

\begin{remark}
Bose creation and annihilation operators are unbounded operators but expressions like $\om \big( a^{\#}_{p_1} (f_1) \dots a^{\#}_{p_n} (f_n) \big)$ are defined for them. So we assume Condition \ref{cond:InitState} holds  also for the Bose case.
\end{remark}

For thermal equilibrium state (Kubo-Martin-Schwinger or KMS state\cite{BratelliRobinson::OperatorAlgabrasQuantumStatMech}) with inverse temperature $\beta$ and chemical potential $\mu$, for Fermi particles
\[
\rho \br{ h } = \frac{ e^{\beta \mu} e^{-\beta h } }{ e^{\beta \mu} e^{-\beta h  } + 1}
\]
and for Bose particles 
\[
\rho \br{ h } = \frac{ e^{\beta \mu} e^{-\beta h } }{ 1 - e^{\beta \mu} e^{-\beta h } }
\]
and Condition~\ref{cond:InitState} holds. Maxwell distribution, which is high-temperature limit of the Bose and Fermi distributions, also fits the considered class (with $\rho \br{h} \sim e^{-\beta h} $). However, with this approach we can consider states other than thermal equilibrium. 

The requirement on the function $\rho$ in Condition  \ref{cond:InitState} is motivated by the application of the stationary phase method \cite{SimonReed::MethodModernMathPhIII}.

The free Hamiltonian without interaction is
\[
H_0 = H^S \otimes \mathbb{I} + \mathbb{I} \otimes  H^R ~.
\]
Let $\Tau^R \br{t}$ and $\Tau^S \br{t}$ be automorphisms  created by the free reservoir and the free system evolutions, respectively, i.e.
\begin{eqnarray*}
\Tau^S (t) X  &=& e^{i H^S t} X e^{-i H^S t} ~, \\
\Tau^R (t) Y  &=& e^{i H^R t} Y e^{-i H^R t}~, \\
\Tau^R (t) ( a^{\#} (f_1) \dots a^{\#} (f_n) )  &=& a^{\#} ( e^{i h t} f_1 ) \dots a^{\#} ( e^{i h t} f_n) ~.
\end{eqnarray*}
Note that for $h = h_1, h_2$ and state $\om$ which satisfy Condition \ref{cond:InitState} we have
\begin{equation} \label{eq:InvariantFreeEvState}
    \om ( \Tau^R (t) ( a^{\# }  (f)  a^{\# } (g) ) ) = \om ( a^{\# }  (f)  a^{\# } (g) ), \quad t \in \RR ~.
\end{equation}
The perturbed evolution of an observable $Y \in \UU^S \otimes \UU^R$ is defined by the Heisenberg equation
\begin{equation} \label{eq:vonNeumann}
\frac{dY}{dt} = i \comm{ H_0 + \lam V}{ Y } ~,
\end{equation}
where  $\lam \in \br{0, 1}$. Equation~\eqref{eq:vonNeumann} defines an automorphism $Y \mapsto \Tau_{\lam} \br{t} Y$.

Partial trace over the reservoir is the linear map $\mathring{\om}: \UU^S \otimes \UU^R \to \UU^S $ defined by the relation $ \mathring{\om} \br{A \otimes B} = A  \times   \om \br{B} $. 

The weak coupling limit is the limit when coupling constant goes to zero, $\lambda\to 0$ and time goes to infinity, $t\to\infty$, but in a related way so that $\lambda^2 t$ remains finite. Then $\tau=\lambda^2 t$ defines a new time scale. 

The reduced dynamics of the system in the weak coupling limit is defined as
\begin{equation} \label{eq:ReducedDynamic}
\wt{\Tau}_{\lam} \br{t} X =  \mathring{\om} \br{ \Tau_{\lam} \br{\lam^{-2}t } \br{X \otimes \mathbb{I}}   },\quad \lam \downarrow 0.
\end{equation}

\section{Dyson Series for the Reduced System Dynamics} \label{sec:DS-definition}
For representing the solution of equation \eqref{eq:vonNeumann} we use standard Dyson series method. Now we consider only Fermi case. Assume $\lam \in \br{0, 1} $ is fixed and introduce the following operator
\[
V \br{t} = \lam e^{-iH_0 t} V e^{i H_0 t} ~,
\]
$V(t)$ is correctly defined on $\DD \otimes \HH^R$, where $\DD$ is the domain appearing in Condition \ref{cond:Interaction}. Consider the following Cauchy problem
\begin{equation} \label{eq:CPvonNeumann}
\left\{ \begin{aligned}
& \frac{d Z \br{t}}{dt} = i \comm{ V \br{t}  }{Z \br{t} } ~, \\
& Z \br{0} = X \otimes \mathbb{I} ~,
\end{aligned} \right.
\end{equation}
where $X \in \UU^S$. Solution of \eqref{eq:CPvonNeumann} can be represented as the series
\begin{equation} \label{eq:DysonSeriesGeneral}
Z \br{t} =  X \otimes \mathbb{I} + \sum_{k=1}^{\infty} i^k \int\limits_{\D_k \br{t} } dt_1 \dots d t_k \comm{ V \br{t_1} }{ \comm{ V \br{t_2}}{ \dots \comm{V \br{t_k} }{ X \otimes \mathbb{I}} \dots} } ~.
\end{equation} 
\begin{lemma} \label{le:DS-converge}
Let Conditions \ref{cond:PermutForHermite} and \ref{cond:Interaction} hold. Then for all $X \in \UU^S$  \\
 1) operator series in the right hand side (r.h.s.) of Eq.~\eqref{eq:DysonSeriesGeneral} is correctly defined and absolutely and uniformly converges on any compact for any argument $y \in \DD \otimes \HH^R$; \\
 2) Operator $Z \br{t}$ defined on the domain $\DD \otimes \HH^R  $ has a bounded closure.
\end{lemma}
\begin{proof}
It is sufficient to prove the statement for  any $X = \br{\cdot, \phi } \psi \in \UU^S $. Consider $Q_j$, $j = 1, \dots, \nu$ from \eqref{eq:Interaction} and $ \eta \in \DD$. Then, with recalling Conditions \ref{cond:PermutForHermite} and \ref{cond:Interaction}, we have
\[
\big[ e^{-iH^S t} Q_j e^{ i H^S t }  ,  X \big] \eta = \br{\eta, \phi} Q_j \psi - \br{\eta, Q_{\sigma \br{j} } \phi } \psi ~, 
\]
where $\sigma$ is permutation from Condition \ref{cond:PermutForHermite}. Since $\DD$ is dense in $\HH^S$, we obtain that the operator $[ e^{-iH^S t} Q_j e^{ i H^S t }, X ]$ can be continued as an bounded operator on whole $\HH^S$. By induction, taking into account that creation and annihilation fermionic operators are bounded (continuous) operators,  we obtain that there exists a constant $C_X$ (which depends on $X$) such that
\[
 \int\limits_{ \D_k \br{t} } dt_1 \dots dt_k \norm{ \comm{ V \br{t_1} }{ \comm{ V \br{t_2}}{ \dots \comm{V \br{t_k} }{ X \otimes \mathbb{I}} \dots} } }  \leq  \frac{ \br{\lam t}^k \br{C_X}^k}{k!} ~. 
\]
The uniform and absolute convergence of the series \eqref{eq:DysonSeriesGeneral} follows from the Weierstrass M-test.
\end{proof}

The formal solution of \eqref{eq:vonNeumann} is  $Y \br{t} = e^{i H_0 t} Z \br{t} e^{-iH_0 t} $. Note that the expression $ H_0 Z \br{t} $ can be undefined when  $H_0$ is an unbounded operator, and then the derivative of $Y \br{t}$ will also be undefined. But even in this case $Y \br{t}$ represents the dynamic in the following sense. Since $\UU^S$ is dense in the algebra of system continuous observables in the strong operator topology, the evolution of a pure state $\phi \mapsto \phi \br{t}$ is explicitly determined (up to a phase factor) by the relation
\[
\br{ \phi, Y \br{t} \phi  } = \br{ \phi \br{t}, Y \phi \br{t}  } \quad  \forall Y \in \UU^S ~.  
\]
We get $\phi \br{t} = e^{-iH_0 t} \psi \br{t}$ and hence $\psi \br{t}$ is represented as the series
\begin{equation} \label{eq:EvolutPureState}
\psi \br{t} = \phi + \sum_{k=1}^{\infty} \br{-i}^k \int_{\D_k \br{t}} dt_1 \dots dt_k V \br{-t_1} \dots V \br{-t_k} \phi ~.
\end{equation}
Proof that series \eqref{eq:EvolutPureState} converges uniformly and absolutely is the same as proof of Lemma \ref{le:DS-converge}. Moreover, $\phi \br{t}$ is a weak solution in the following sense. For an arbitrary $\eta \in \DD \otimes D ( H^R) $ consider the function \cite{SimonReed::MethodModernMathPhII}
\[
f_{\eta} \br{t} = \br{ \phi \br{t}, \eta } ~.
\]
Then function $f_{\eta} \br{t}$ is differentiable for any $t$ and
\[
\frac{d f_{\eta} \br{t}}{dt} = \br{ \phi \br{t}, i \br{H_0 + \lam V} \eta } ~.
\]

Since $\mathring{\om}$ and $\Tau^R $ are continuous maps, using Lemmas \ref{le:DS-converge} and \eqref{eq:InvariantFreeEvState} we obtain the following form of the Dyson series for the reduced system dynamic \eqref{eq:ReducedDynamic}
\begin{eqnarray} 
\label{eq:DysonSeriesReducedDynExpandedForm}
\wt{\Tau}_{\lam} \br{t} X &=& \mathring{\om} \Big( \big( \Tau^S (\lam^{-2} t) \otimes \mathbb{I}  \big) \big( \mathbb{I} \otimes  \Tau^R (\lam^{-2} t)  \big) \big(  X \otimes \mathbb{I}  \nonumber \\ 
 &&+  \sum_{k=1}^{\infty} i^k \int\limits_{\D_k \br{\lam^{-2} t} } dt_1 \dots dt_k \left[ V \br{t_1},  \dots \left[ V \br{t_k}, X \otimes \mathbb{I}  \right] \dots \right]  \big)  \Big)    \nonumber \\ 
&=& \Tau^S (\lam^{-2} t) \Big( X  + \sum_{k=1}^{\infty} i^k \int\limits_{\D_k \br{\lam^{-2} t} } dt_1 \dots dt_k  \mathring{\om} \big(  \left[ V \br{t_1},  \dots \left[ V \br{t_k}, X \otimes \mathbb{I}  \right] \dots \right] \big) \Big)~.
\end{eqnarray}

Let us introduce some notations. Define subset $S'(k) \subset S(k)$  for $k=1,2,3,\dots$ recursively. The base is $S'(1) = S(1) = \{ \text{id} \}$. Consider two operators $S_{\pm}: S(k) \to S(k+1) $
\[
    \left\{ \begin{aligned}
    &S_{+} \br{\pi} = \br{1, \pi \br{1}+1, \pi \br{2}+1, \dots, \pi \br{k}+1 }, \\
    &S_{-} \br{\pi} = \br{ \pi \br{1}+1, \pi \br{2}+1, \dots, \pi \br{k}+1, 1 },
    \end{aligned} \right. ~\pi \in S \br{k} ~.
\]
By definition, let $S' \br{k+1} = S_{+} \br{ S' \br{k}} \cup S_{-} \br{ S' \br{k}}$. Note that $\abs{S' \br{k}} = 2^{k-1}$. 

By $1_k \leq \gamma \leq \nu 1_k$ for a multi-index $\gamma = (\gamma_1, \dots , \gamma_k)$ we mean that $ 1 \leq \gamma_j \leq \nu$ for $j = 1, \dots, k$. Define the operator $F_{\gamma, \pi} \br{t_1, \dots, t_k}$ for $\pi \in S' \br{k}$, $1_k \leq \gamma \leq \nu 1_k$  in the reservoir Hilbert space by
\[
F_{\gamma, \pi } \big( t_1, \dots , t_k \big) = F_{\gamma_{\pi(1) }} \big( t_{\pi (1) } \big) \dots F_{\gamma_{\pi (k) }} \big( t_{\pi \br{k}} \big) ~,
\]
where $F_{j} ( \tau ) = \Tau^{R} ( -\tau ) F_j $. The operator $Q_{\gamma, \pi} (t)$ is denoted recursively by
\begin{eqnarray*}
Q_{\gamma, \text{id} } [ A ] &=&  [ Q_{\gamma} , A ], ~ \gamma = 1, 2, \dots, \nu ~.\\
Q_{\gamma, \pi}   [ A ]  &=& \left\{  \begin{aligned}
& Q_{\gamma_1} \times Q_{ \gamma', \pi' }  [ A ],~ \pi = S_{+} \br{\pi'} ~, \\
& -Q_{ \gamma', \pi' }  [A ] \times Q_{\gamma_1},~ \pi = S_{-} \br{\pi'} ~,
\end{aligned} ~ \right. \quad  \gamma = (\gamma_1, \gamma') ~.
\end{eqnarray*}
With these notations we rewrite series \eqref{eq:DysonSeriesReducedDynExpandedForm} in WCL as follows
\begin{eqnarray} \label{eq:DysonSeriesReducedDynamic}
    \wt{\Tau}_{\lam} (t) X &=& \Tau^S (\lam^{-2} t)    \Big( X \nonumber \\
    &&+ \sum_{k=1}^{\infty} i^k \sum_{1_k \leq \gamma \leq \nu 1_k}  \sum_{ \pi \in S' \br{k}  }  \int\limits_{\D_k \br{\lam^{-2}t}}  dt_1 \dots dt_k Q_{\gamma, \pi} \left[ X \right] \om  \big( \lam^{k} F_{\gamma, \pi} \br{t_1, t_2, \dots, t_k}    \big)  \Big) ~.
\end{eqnarray}
The series \eqref{eq:DysonSeriesReducedDynamic} for a boson reservoir is understood as a formal term by term correspondence with the corresponding terms of the series \eqref{eq:DysonSeriesReducedDynExpandedForm}. We prove that the series \eqref{eq:DysonSeriesReducedDynamic} converges for both case of Fermi and Bose reservoirs.

\section{Reservoir Correlation Functions in the Weak Coupling Limit}\label{Sec3:Statistics}
In this section, we derive the limit of correlation functions of the reservoir which appear in the Dyson series. Details of the proof are provided in the Appendix.

Let $\pi \in S' \br{k}$ and $\gamma$ be a multi-index such that  $1_k \leq \gamma \leq \nu 1_k$. Consider the following complex-valued function $\Om_{\gamma, \pi}$ of real argument $\lam \in \br{0, 1}$
\begin{equation} \label{eq:DysonSeriesReservoirIntegral}
\Om_{\gamma, \pi} \br{\lam} = \lam^{k} \int\limits_{ \D_k \br{ \lam^{-2} t}} \om \br{F_{\gamma, \pi} \br{t_1, \dots, t_k}  }  dt_1 \dots dt_k ~.
\end{equation}
In this section we derive the limit $\Om_{\gamma, \pi} \br{\lam}, \lambda \downarrow 0 $ for fixed $t$.
Taking into account the Condition \ref{cond:InitState}, we immediately obtain that if $\pi \in S' \br{2k+1}$ then $\Om_{\gamma, \pi} \br{\lam} = 0$. Thus we need to consider only $\pi \in S' \br{2k}$.

We will use the functions
\begin{eqnarray*}
 \chi \br{\tau} &=& \frac{1}{ \br{1 + \abs{\tau} }^\frac{3}{2} } ~, \\
 G_{k, p} \br{\lam} &=&  \lam^{2p} \int\limits_{S_k \br{\lam^{-2}t}} dt_1 \dots dt_k \chi \br{t_1} \dots \chi \br{t_k} \br{t_1 + \dots + t_k}^{p}, ~ k \in \NN, p \geq 1 ~.
\end{eqnarray*}
We have the following estimates.
\begin{lemma} \label{le:EstimateIntegral}
For all  $\lam \in \br{0, 1}$ the following inequalities hold
\begin{eqnarray*}
\lam^{2p} \int\limits_0^{ \lam^{-2}t } \tau^p \chi \br{\tau} d \tau & \leq & 2  t^{p -\frac{1}{2}} \lam ~, \\
G_{k, p} \br{\lam} & \leq &  k 2^{k+p}  t^{p-\frac{1}{2}} \lam~.
\end{eqnarray*}
\end{lemma}
Proof of Lemma \ref{le:EstimateIntegral} is given in Appendix \ref{app:ProofLemma}.

\begin{lemma} \label{le:fugacityIneq}
Suppose that $f, g \in \mathfrak G$ and that $\om$ satisfies Condition \ref{cond:InitState}. Then there exists a constant $C$ such that 
\begin{equation} \label{eq:fugacityIneq}
\Big| \om  \big(  a^{\dagger}  ( e^{it_1 h} f ) a ( e^{i t_2 h}  g )   \big) \Big| , \Big| \om \big( a  ( e^{i t_2 h}  g )  a^{\dagger} ( e^{it_1 h} f ) \big) \Big| \leq C  \chi ( t_1 - t_2 )~,
\end{equation}
where $h = h_1, h_2$.
\end{lemma}
\begin{proof}  First consider the case of Fermi particles and $h = h_2$. By definition we have 
\[
    \om \big( a^{\dagger} ( e^{i t_1 h} f )  a  ( e^{i t_2 h} g) \big) =   \int\limits_{\RR^3}  \rho \big( |k|^2 \big) f(k) \overline{g(k)} e^{ i(t_1-t_2) |k|^2  } dk~.
\]
Using the anti-commutation relation for Fermi creation and annihilation operators, we get 
\[
\om \big( a (e^{i t_1 h} f)  a^{\dagger}  (e^{i t_2 h} g) \big) = \big( e^{i t_2 h} g,  e^{i t_1 h} f \big) - \om \big( a^{\dagger} ( e^{i t_1 h} f)  a (e^{i t_2 h} g) \big) ~.
\]
Then inequality \eqref{eq:fugacityIneq} follows from the stationary phase method \cite{SimonReed::MethodModernMathPhIII, Fedoryuk::SaddlePointMethod} with respect to the integration function $  {f} (k) \overline{ g (k) }  \rho ( \abs{ k}^2 )$ which belongs to $\mathfrak G$.

Consider case $h = h_1$.  By definition we obtain
\[
    \om \big( a^{\dagger} ( e^{i t_1 h} f )  a ( e^{i t_2 h} g)  \big) = \int\limits_{\RR^3}   f (k) \overline{  g  (k) } \rho (\abs{ k} ) e^{ i (t_1 - t_2) \abs{ k}  } dk ~.
\]
In spherical coordinates with integration over angle variables
\begin{equation} \label{eq:tempIntegral}
    \om \big( a^{\dagger}  ( e^{i t_1 h} f )  a  ( e^{i t_2 h} g ) \big) = \int\limits_{0}^{\infty}   y (r) e^{ i ( t_1 - t_2) r  } dr ~,
\end{equation}
where $ y(r) $ is a finite infinitely differentiable function. Applying Theorem XI.14 from Ref. \onlinecite{SimonReed::MethodModernMathPhIII}  to the integral \eqref{eq:tempIntegral} we obtain the required. The case of Bose particle can be considered similarly. 
\end{proof}

\begin{remark}
 In the case $h = h_1$ function $ c ( \D t ) = \om ( a^{\dagger} ( e^{i h \D t} f)  a (g) )$ decreases faster than any inverse  polynomial\cite{SimonReed::MethodModernMathPhIII}. But the above estimate is sufficient for our subsequent analysis. Informally speaking, the main result of the paper requires for the function $c ( \D t )$ to decrease fast enough to be integrable on $\br{- \infty, +\infty}$.
\end{remark}

Denote $ C = \max \big\{ C ( f_{ip}, g_{jp}, \om), 1 \leq i, j \leq \nu, p = 1, 2 \big\} $,
where $C ( f_{ip} , g_{jp}, \om )$ is constant from Lemma \ref{le:fugacityIneq} for functions $f_{ip}, g_{jp}$ and state $\om$.

\begin{lemma} \label{le:ZeroFD}
Assume that Condition \ref{cond:Formfactors} holds and $\om$ satisfies Condition \ref{cond:InitState}. Let $(n_j, m_j)$ be  a partitioning $\left\{1, 2, \dots, 2k \right\}$ into pairs.  Consider the functions $\{ \phi_j \}_{j=1}^{2k}$,  $\phi_j \in  \{ f_{ip}, g_{ip}, 1 \leq i \leq \nu, p = 1, 2 \}$ and 
\begin{eqnarray}  \label{eq:lemmaFDGeneralInt}
\Phi (\lam) &=&  \lam^{2k} \int\limits_{\D_{2k} ( \lam^{-2}t) } \om \big( a^{\#} ( e^{-i t_{n_1} h } \phi_{1}  )  a^{\#} ( e^{- i t_{m_1} h } \phi_{2}  )  \big)  \om \big( a^{\#} ( e^{-i t_{n_2} h } \phi_{3}  )  a^{\#} ( e^{-i t_{m_2} h } \phi_{4}  )  \big) \dots \nonumber  \nonumber \\
&&   \dots \times \om \big( a^{\#} ( e^{-i t_{n_k} h } \phi_{(2k-1)}  )  a^{\#}  ( e^{-i t_{m_k} h } \phi_{2k}  )  \big) dt_1 \dots dt_{2k} ~.
\end{eqnarray}
If $ | n_j - m_j | > 1$ for some $j$, then for any $\lam \in \br{0, 1}$ the following estimate holds:
\[
\abs{\Phi \br{\lam} } \leq \frac{ \br{16 C}^k t^{k - \frac{1}{2}}  }{  \br{k-1}! } \lam ~.
\]
\end{lemma}
\begin{proof}
Without loss of generality we assume that  $n_j - m_j > 1$ for some $j$. Then there exists $j^*$ such that either $m_j < n_{j^*} < n_j$ or $m_j < m_{j^*} < n_j$. Consider the case $m_j < n_{j^*} < n_j$; another case can be considered similarly.  Make change of the variable in the integral \eqref{eq:lemmaFDGeneralInt}.  Consider an array $M = \left\{ m_i, i \neq j^* \right\} $ and let $m'_1, m'_2, \dots m'_{k-1}$ be elements of $M$ sorted in the increasing order. Introduce the new variables  $s_i = t_{ m'_i }$ for $i < k$, $s_k = t_{n_{j^*}}$, and  $\tau_i = t_{n_i} - t_{m_i}$. Each row of the Jacobi matrix corresponding to this change of variable contains at most two $\pm 1$ with all other zero elements. So the absolute value of the determinant of the Jacobi matrix is at most $2^{2k}$. Using  Lemma \ref{le:fugacityIneq}  we get 
\begin{eqnarray*}
 \abs{ \Phi \br{\lam} } &\leq & \lam^{2k} 2^{2k} C^k \int\limits_{ \br{-\lam^{-2}t, \lam^{-2}t}^k } d \tau_1 \dots d \tau_k \int\limits_{ \D_{k-1} \br{\lam^{-2} t} } d s_1 \dots ds_{k-1} \int\limits_{ t_{n_j} }^{t_{m_j}} d s_k \chi \br{\tau_1} \dots \chi \br{\tau_k}  \\
    &\leq & \lam^{2k} 4^k C^k \frac{ \br{\lam^{-2}t}^{k-1}  }{\br{k-1}!} \int\limits_{ \br{-\lam^{-2}t, \lam^{-2}t}^k } | \tau_j | \chi \br{\tau_1} \dots \chi \br{\tau_k}   d \tau_1 \dots d \tau_k   \\
    &\leq & \frac{ \br{4C}^k t^{k-1} }{ \br{k-1}! } \left( \int\limits_{-\infty}^{+\infty} \chi \br{\tau} d \tau  \right)^{k-1} \br{ 2 \lam^2 \int\limits_0^{\lam^{-2} t} \tau \chi \br{\tau} d \tau  }~.
\end{eqnarray*}
Finally, taking into account Lemma \ref{le:EstimateIntegral} we obtain
\[
    \abs{ \Phi \br{\lam} } \leq \frac{ \br{4 C}^k t^{k - 1} }{  \br{k-1}!} 4^{k-1} \br{4 \lam t^{\frac{1}{2}}} = \frac{ \br{16 C}^k t^{k - \frac{1}{2}}  }{ \br{k-1}!} \lam ~.
\]
\end{proof}
In a similar way we get the following estimate.
\begin{lemma} \label{le:FD-Estimate}
Under the conditions of Lemma \ref{le:ZeroFD}, one has
\[
\abs{ \Phi \br{\lam} } \leq \frac{ \br{16C t}^k }{ k! }~.
\]
\end{lemma}
\begin{proof}
Making a similar change of the variable one gets
\begin{eqnarray*}
 \abs{\Phi \br{\lam}}  &\leq &  \lam^{2k} 2^{2k} C^k \int\limits_{ \br{-\lam^{-2}t, \lam^{-2}t}^k } d \tau_1 \dots d \tau_k \int\limits_{ \D_{k} \br{\lam^{-2} t} } d s_1 \dots ds_{k}  \prod_{i=1}^k \chi \br{\tau_k}  \\ 
    &\leq &  \lam^{2k} \br{4C}^k \frac{ \br{\lam^{-2}t}^k }{k!} \left( \int\limits_{-\infty}^{+\infty} \chi \br{\tau} d \tau \right)^k \leq \frac{ \br{16 C t}^k }{k!} ~.
\end{eqnarray*}
\end{proof}

\begin{lemma} \label{le:FD-Summ}
Assume that Condition \ref{cond:Formfactors} holds and $\om$ satisfies Condition \ref{cond:InitState}. Consider $\pi \in S' \br{2k}$, functions $ \{ \phi_j \}_{j=1}^{2k}$,  $\phi_j \in \left\{ f_{ip}, g_{ip}, 1 \leq i \leq \nu, p =1, 2 \right\}$, and  a monomial
\[
    P \br{ t_1, t_2, \dots, t_{2k}} = a^{\#} \big( e^{- i t_{\pi \br{1}} h } \phi_1  \big) a^{\#} \big( e^{ -i t_{\pi \br{2}} h } \phi_2  \big) \dots a^{\#} \big( e^{ -i t_{\pi \br{2k}} h } \phi_{2k}  \big) ~.
\]
Let $\sigma \in S \br{2k}$ be defined as 
\[
    \sigma \br{2j-1} = \min \left\{ \pi^{-1} \br{2j-1}, \pi^{-1} \br{2j}  \right\}, \sigma \br{2j} = \max  \left\{ \pi^{-1} \br{2j-1}, \pi^{-1} \br{2j}  \right\} ~, 
\]
let $\tau_j = t_{2j-1} - t_{2j}$ and 
\[
\psi_j ( \tau_j  )  = \om \big( a^{\#} ( e^{ -i t_{\pi \br{ \sigma \br{2j-1} }} h } \phi_{\sigma \br{2j-1}}   )  a^{\#} ( e^{- i t_{\pi \br{ \sigma \br{2j} }} h } \phi_{\sigma \br{2j}}   )  \big) ~. 
\]
Then for $\lam \in \left(0 , 1 \right)$ holds
\begin{equation} \label{eq:MonomialLimitFormula}
    \lam^{2k} \int\limits_{\D_{2k} \br{\lam^{-2} t} } d t_1 \dots d t_{2k} \om \br{ P \br{t_1, t_2, \dots, t_{2k}}  } = \frac{t^k}{k!} \prod_{j=1}^k \int\limits_0^{\infty} \psi_j \br{t} dt + R_k \br{\lam}
\end{equation}
and $R_k \br{\lam}$ can be estimated as
\begin{equation} \label{eq:EstimatedRes}
    \abs{ R_k \br{\lam} } \leq   \frac{ Ak \br{C_1 t}^k }{\sqrt{t}} \lam ~,
\end{equation}
where $A, C_1$ are some constants.
\end{lemma}
\begin{proof}
First, note that $\mathop{sgn} \br{\sigma} = 1$. Then the proof can be done by induction on $k$. Consider the case $k = 1$. $S' \br{2} = \left\{ \br{1,2}, \br{2, 1} \right\}$ and $\sigma = \br{1, 2}$, $\mathop{sgn} \br{\sigma}=1$. Consider $\pi \in S' \br{2 \br{k+1} } $. There are four possible cases:
\begin{eqnarray*}
    && 1)~ \pi = S_{+} \br{S_{+} \br{\pi'} },\qquad 2)~ \pi = S_{+} \br{S_{-} \br{\pi'} }, \\
    && 3)~ \pi = S_{-} \br{S_{+} \br{\pi'} },\qquad 4)~ \pi = S_{-} \br{S_{-} \br{\pi'} }.
\end{eqnarray*}
Let $\sigma'$ correspond to $\pi'$ and $\sigma$ correspond to $\pi$. By definition for the above cases we have
\begin{eqnarray*}
    && 1)~ \sigma = \br{ 1, 2, \sigma' \br{1}+2, \dots, \sigma' \br{2k}+2 },\qquad 2)~ \sigma = \br{1, 2k+2, \sigma' \br{1}+1, \dots, \sigma' \br{2k}+1}, \\
    && 3)~ \sigma = \br{1, 2k+2, \sigma' \br{1}+1, \dots, \sigma' \br{2k}+1},\qquad 4)~ \sigma = \br{ 2k+1, 2k+2, \sigma' \br{1}, \dots, \sigma' \br{2k} }.
\end{eqnarray*}
We obtain that $\mathop{sgn} \br{\sigma} = \mathop{sgn} \br{\sigma'} $ and therefore $\mathop{sgn} \br{\sigma} = 1$.

According to Lemma \ref{le:ZeroFD} only one partitioning in \eqref{eq:quasi-free} for \eqref{eq:MonomialLimitFormula} is non-zero in the limit. By definition, this  partitioning is determined by the permutation $\sigma$. Other $k!-1$ summands can be estimated using lemma \ref{le:ZeroFD}. As a result we obtain
\begin{eqnarray*}
&& \lam^{2k} \int\limits_{\D_{2k} \br{\lam^{-2} t} } d t_1 \dots d t_{2k} \om \br{ P \br{t_1, t_2, \dots, t_{2k}}  }  \\
&& = \lam^{2k} \int\limits_{\D_{2k} \br{\lam^{-2}t} } dt_1 \dots dt_{2k} \prod_{j=1}^k \om \big( a^{\#} ( e^{ -i t_{\pi \br{ \sigma \br{2j-1} }} h } \phi_{\sigma \br{2j-1}}   )  a^{\#} ( e^{ -i t_{\pi \br{ \sigma \br{2j} }} h } \phi_{\sigma \br{2j}}   )   \big)  + R_{k1} \br{\lam}
\end{eqnarray*}
and 
\begin{equation} \label{eq:UpperBoundFirstTerm}
    \abs{R_{k1} \br{\lam}} \leq   k \frac{ \br{16 C t}^k  }{  \sqrt{t} } \lam ~.
\end{equation}

Consider change of the variable $\tau_j = t_{2j-1} - t_{2j} $, $s_j = t_{2j}$. Absolute value  of the Jacobian   for this change of the variable is $1$. Then we have
\begin{eqnarray*}
&& \lam^{2k} \int\limits_{\D_{2k} \br{\lam^{-2} t} } d t_1 \dots d t_{2k} \om \br{ P \br{t_1, t_2, \dots, t_{2k}}  }  \\ 
&& = \lam^{2k} \int\limits_{ S_k \br{\lam^{-2} t} } d \tau_1 \dots d \tau_k \int\limits_0^{ \lam^{-2} t - \br{ \tau_1 + \dots + \tau_k }  } d s_k  \dots \int\limits_{s_2 + \tau_2}^{ \lam^{-2}t - \tau_1  } d s_1 \psi_1 \br{\tau_1} \dots \psi_k \br{\tau_k} +  R_{k1} \br{\lam} ~.
\end{eqnarray*}
Note that 
\[
    \int\limits_0^{ \lam^{-2} t - \br{ \tau_1 + \dots + \tau_k }  } d s_k \int\limits_{ s_k + \tau_k }^{ \lam^{-2}t - \br{\tau_1 + \dots + \tau_{k-1}}  } d s_{k-1} \dots \int\limits_{s_2 + \tau_2}^{ \lam^{-2}t - \tau_1  } d s_1 = \frac{ \br{ \lam^{-2}t - \br{\tau_1 + \dots + \tau_k}  }^k   }{k!} ~.
\]
Indeed, making change of variable $ s'_k =  s_k $, $s'_{k-1} = s_{k-1}-\tau_k,$ $\dots ,$  $s'_1 = s_1 - \br{\tau_1 + \dots + \tau_k}$ we get
\begin{eqnarray*}
&& \int\limits_0^{ \lam^{-2} t - \br{ \tau_1 + \dots + \tau_k }  } d s_k \int\limits_{ s_k + \tau_k }^{ \lam^{-2}t - \br{\tau_1 + \dots + \tau_{k-1}}  } d s_{k-1} \dots \int\limits_{s_2 + \tau_2}^{ \lam^{-2}t - \tau_1  } d s_1  \\
&&  = \int\limits_0^{ \lam^{-2} t - \br{\tau_1 + \dots + \tau_k} } d s'_k \int\limits_{s'_k}^{ \lam^{-2} t - \br{\tau_1 + \dots + \tau_k} } d s'_{k-1} \dots \int\limits_{s'_2}^{\lam^{-2} t - \br{\tau_1 + \dots + \tau_k}} d s'_1 = \frac{ \br{\lam^{-2} t - \br{\tau_1 + \dots + \tau_k}}^k   }{k!} ~.
\end{eqnarray*}
We obtain
\begin{eqnarray*}
&& \lam^{2k} \int\limits_{\D_{2k} \br{\lam^{-2} t} } d t_1 \dots d t_{2k} \om \br{ P \br{t_1, t_2, \dots, t_{2k}}  }  \\
 &&   = \frac{\lam^{2k}}{k!} \int\limits_{S_k \br{\lam^{-2} t}} d \tau_1 \dots d \tau_k \psi_1 \br{\tau_1} \dots \psi_k \br{\tau_k} \br{ \lam^{-2} t - \br{\tau_1 + \dots \tau_k}  }^k + R_{k1} \br{\lam} ~.
\end{eqnarray*}
Now applying the binomial expansion we get
\begin{eqnarray*}
&& \frac{\lam^{2k}}{k!} \int\limits_{S_k \br{\lam^{-2} t}} d \tau_1 \dots d \tau_k \psi_1 \br{\tau_1} \dots \psi_k \br{\tau_k} \br{ \lam^{-2} t - \br{\tau_1 + \dots \tau_k}  }^k  \\
&& = \frac{t^k}{k!} \sum_{p=1}^k \binom{k}{p} \br{-1}^p  t^{-p} \lam^{2p} \int\limits_{S_k \br{\lam^{-2} t}} d \tau_1 \dots d \tau_k \psi_1 \br{\tau_1} \dots \psi_k \br{\tau_k} \br{\tau_1 + \dots + \tau_k}^p  \\
&&\hspace{.4cm}+\frac{t^k}{k!} \int\limits_{S_k \br{\lam^{-2} t}} d \tau_1 \dots d \tau_k \psi_1 \br{\tau_1} \dots \psi_k \br{\tau_k} ~.
\end{eqnarray*}
The last summand is in general non-zero in the limit. For estimating the sum for $p\geq 1$, combining Lemmas \ref{le:EstimateIntegral}, \ref{le:fugacityIneq} gives
\begin{eqnarray}  \label{eq:UpperBoundSecondTerm}
  \abs{R_{k2} \br{\lam} } &= &  \Big| \frac{t^k}{k!} \sum_{p=1}^k \binom{k}{p} \br{-1}^p  t^{-p} \lam^{2p} \int\limits_{S_k \br{\lam^{-2} t}} d \tau_1 \dots d \tau_k \psi_1 \br{\tau_1} \dots \psi_k \br{\tau_k} \br{\tau_1 + \dots + \tau_k}^p \Big| \nonumber \\ 
  &\leq& \frac{ \br{Ct}^k}{k!} \sum_{p=1}^k \binom{k}{p} t^{-p}  G_{k, p} \br{\lam} \leq \frac{ \br{2C t}^k }{ \br{k-1}! \sqrt{t} } \sum_{p=1}^k \binom{k}{p} 2^p \lam \leq \frac{ \br{6C t}^k }{ \br{k-1}! \sqrt{t} } \lam ~.
\end{eqnarray}

Taking the limit $\lam \downarrow 0$ we get
\[
    \int\limits_{S_k \br{\lam^{-2} t} } d \tau_1 \dots d \tau_k \psi_{1} \br{\tau_1} \dots \psi_{k} \br{\tau_k} = \prod_{j=1}^k \int\limits_0^{+\infty} \psi_j \br{\tau} d \tau + R_{k3} \br{\lam} ~.
\]
Each function $\psi_j \br{\tau}$ is continious and using Lemma \ref{le:fugacityIneq} and Cauchy criterion it is easy to see that the integral converges uniformly for $\lam \in \br{0, 1}$. We  estimate $R_{k3} \br{\lam} $ using Lagrange theorem. For this, introduce the functions
\[
    \Psi_n \br{y} = \int\limits_{S_n \br{y} } d \tau_1 \dots d \tau_n \psi_1 \br{ \tau_1 } \dots \psi_n \br{\tau_n},\quad 1 \leq n \leq k ~.
\]
Then let us prove that there exists a constant $c$ such that 
\begin{equation} \label{eq:UpperBoundDiff}
    \Big| \frac{ d \Psi_n \br{y} }{ d y } \Big| \leq \frac{c^n}{ \br{1+y}^\frac{3}{2} }, \quad y > 0 ~.
\end{equation}
Inequalities \eqref{eq:UpperBoundDiff} are proven by induction. For $n=1$
\[
    \frac{ d \Psi_n \br{y} }{ d y } = \psi_{1} \br{ y} ~.
\]
Using Lemma \ref{le:fugacityIneq}, we obtain that the base of the induction holds for any $c>C$. Consider $n > 1$
\[
    \Psi_n \br{y} = \int\limits_{0}^y \psi_n \br{ \tau_n}  \Psi_{n-1} \br{y - \tau_n} d \tau_n ~.
\]
Combining Leibniz integral rule and obvious identity $\Psi_{n-1} \br{0} = 0$, we obtain that
\[
  \frac{ d \Psi_{n} }{dy} \br{y} = \int\limits_0^y \psi_n \br{\tau} \frac{d \Psi_{n-1}}{d y} \br{y-\tau} d \tau.
\]
By the induction assumption and  Lemma \ref{le:fugacityIneq} one gets 
\[
    \Big| \frac{ d \Psi_{n} }{ dy }  \Big| \leq  c^{n-1} C \int\limits_0^y \frac{ 1   }{\br{1+\tau}^{\frac{3}{2}} } \frac{ 1   }{\br{1+ y - \tau}^{\frac{3}{2}} } d \tau = c^{n-1} C \frac{ 4 y }{ \sqrt{1+y} \br{y+2}^2 } \leq \frac{c^n}{ \br{1+y}^\frac{3}{2} }
\]
and the inequality is satisfied for any sufficiently large $c$. Note that the constant $c$ can be chosen uniformly for all possible $\phi_j$.

Using Lagrange theorem, we obtain
\begin{equation} \label{eq:UpperBoundThirdTerm}
    \abs{ R_{k3} \br{\lam} } \leq \Big| \frac{\Psi_{k} \br{\lam^{-2} t} }{ d \lam}  \Big| \lam  \leq \frac{2 c^k}{\sqrt{t}} \lam, \quad \lam \in \br{0, 1} ~.
\end{equation}

Finally, it gives
\[
        \lam^{2k} \int\limits_{\D_{2k} \br{\lam^{-2} t} } d t_1 \dots d t_{2k} \om \br{ P \br{t_1, t_2, \dots, t_{2k}}  } = \frac{t^k}{k!} \prod_{j=1}^k \int\limits_0^{\infty} \psi_j \br{t} dt  + R_{k1} \br{\lam} + R_{k2} \br{\lam} + R_{k3} \br{\lam} ~.
\]
Summing \eqref{eq:UpperBoundFirstTerm}, \eqref{eq:UpperBoundSecondTerm}, and \eqref{eq:UpperBoundThirdTerm}, one gets the estimate
\[
    \abs{ R_k \br{\lam} } = \abs{ R_{k1} \br{\lam} + R_{k2} \br{\lam} + R_{k3} \br{\lam} } \leq \frac{ Ak \br{C_1 t}^k }{\sqrt{t}} \lam, \quad \lam \in \br{0, 1} ~,
\]
where constants $A, C_1$ can be set uniformly of choosing $\phi_j$.
\end{proof}

\begin{theorem} \label{th:FD-RecurRelation}
Assume that Conditions \ref{cond:Formfactors} and  \ref{cond:InitState} hold. Define
\begin{eqnarray*}
    a_{ij} & =& \int\limits_0^{\infty} \om \br{ F_i \br{t} \times F_j }  dt ~, \\
    b_{ij} & =& \int\limits_0^{\infty} \om \br{ F_j \times F_i \br{t} } dt ~.
\end{eqnarray*}
One has the following.
\begin{itemize}
\item If $\pi \in S' \br{2k+1}$ then for any multi-index $1_{2k+1} \leq \gamma \leq \nu 1_{2k+1}$ and $\lam \in \br{0, 1}$ one has
    \[
    \Om_{\gamma, \pi} \br{\lam} = 0 ~.
    \]
\item For $\pi \in S' \br{2k}$ one has
    \[
    \Om_{ \gamma, \pi  } \br{\lam} = U \br{\gamma, \pi} + R_{\gamma, \pi} \br{\lam} ~.
    \]  
Here term $R_{\gamma, \pi}$ satisfies
    \[
    \abs{ R_{\gamma, \pi} \br{\lam}  } \leq  \frac{ k A  \br{ C' t}^k }{ \sqrt{t} } \lam
    \]
with some constants $A, C'$ and $\lam \in \br{0, 1}$. Term $U \br{\gamma, \pi}$ satisfies the following recurrence relation:
\begin{itemize}
    \item For $k=1$,
    \begin{eqnarray*}
         U \br{ \br{i, j}, \br{1, 2} } &=& t a_{ij} ~, \\
        U \br{ \br{i, j}, \br{2, 1} } &=& t b_{ij} ~.
    \end{eqnarray*}
    \item For $\pi \in S \br{2 \br{k+1} } $ if either $\pi = S_{+} \br{S_{+} \br{\pi'}  }  $ or $\pi = S_{+} \br{S_{-} \br{\pi'}}$ then
   \[
    U \br{ \br{i, j, \gamma'}, \pi   } = \frac{ t a_{ij }}{ k+1 }  U \br{\gamma', \pi'} ~,   
    \]
    else
    \[
        U \br{ \br{i, j, \gamma'}, \pi   } = \frac{ t b_{ij }}{ k+1 }  U \br{\gamma', \pi'} ~.
    \]
\end{itemize}
\end{itemize}
\end{theorem}
\begin{proof}
The statement for $\pi \in S' \br{2k+1}$  follows directly from that state $\om$ is quasi-free.  

The statement for $\pi \in S' \br{2k}$ can be proved by induction on $k$.  Consider $\pi \in S' \br{2}$. For the permutation $\pi = \br{1, 2} $ we have
\begin{eqnarray*}
    \Om_{ \br{i, j}, \br{1, 2} } \br{\lam} & = & \lam^2 \int\limits_{\D_2 \br{\lam^{-2} t} } \om \br{ F_{i} \br{t_1} F_j \br{t_2} }  dt =   \lam^2 \int\limits_{\D_2 \br{\lam^{-2} t} } \om \br{ F_i \br{ t_1 - t_2 } F_j  } dt \,\left[ \tau := t_1 - t_2 \right] \\ 
     & = &  \int_0^{\lam^{-2} t} \br{ t - \lam^2 \tau }  \om \br{ F_i \br{ \tau } F_j } d \tau = t a_{ij} - \lam^2 \int\limits_0^{\lam^{-2} t} \tau \om \br{ F_i \br{ \tau } F_j } d \tau ~.
\end{eqnarray*}
Taking into account Lemmas \ref{le:fugacityIneq} and \ref{le:EstimateIntegral} we obtain the base case.

Consider $\pi \in S' \br{2k}$, $1_{2k} \leq \gamma \leq \nu 1_{2k} $, $\gamma = \br{i, j, \gamma'}$. Let $\pi = S_{+} \br{ S_{+} \br{\pi'} }$; other cases can be considered similarly. By definition, we have
\[
    F_{\gamma, \pi } \br{t_1, \dots, t_{2k} }  = F_{ i } \br{t_1} F_j \br{t_2} F_{\gamma', \pi'} \br{t_3, \dots, t_{2k}} ~.
\]
Here $F_{\gamma, \pi }$ consists of $4^{2k}$ terms, where each term is a monomial of creation and annihilation operators. In each monomial we take the limit $\lam \downarrow 0$ using Lemma \ref{le:FD-Summ}. Taking into account \eqref{eq:MonomialLimitFormula} we obtain
\[
    U \br{ \br{i, j, \gamma'}, \pi  } = \frac{ t }{k} \bigg( \int\limits_0^{\infty} \om \br{ F_i \br{t} F_j } dt \bigg) U \br{ \gamma' , \pi' } + R_{ \gamma, \pi } \br{\lam} ~,  
\]
where $R_{ \gamma, \pi } \br{\lam}$ can be estimated using \eqref{eq:EstimatedRes} as
\[
\abs{R_{ \gamma, \pi } \br{\lam}} \leq 4^{2k} \frac{ Ak \br{C_1 t}^k }{\sqrt{t}} \lam ~.
\]
\end{proof}

Expressions for $a_{ij}$ and $b_{ij}$ can be written in another form by applying stationary phase method. For one-particle Hamiltonian $h = h_q$, $q = 1, 2$ we have 
\[
\int_0^{+\infty} \om \big( a^{\dagger} ( e^{ \pm i t h_q} f ) a ( g )   \big) dt = \lim_{T \to \infty} \int\limits_0^T dt \int\limits_{\RR^3} d  k  f  ( k ) \overline{ g (k ) } \rho ( |k|^q ) e^{ \pm i t |k|^q } ~.
\]
With Fubini's theorem 
\[
\int\limits_0^T dt \int\limits_{\RR^3} d  k f (k) \overline{ g (k) } \rho ( |k|^q ) e^{ \pm i t |k|^q } = \int\limits_{\RR^3}  dk  f (k) \overline{ g (k) } \rho ( |k|^q) \frac{ e^{\pm i T |k|^q } -1 }{ \pm i |k|^q } ~.
\]
Let us switch to spherical coordinates
\[
    \left\{ \begin{aligned}
     k_1 &= r \sin \br{ \theta } \cos \br{ \phi }, \\
     k_2 &= r \sin \br{ \theta } \sin \br{ \phi }, \\
     k_3 &= r \cos \br{ \theta }.
    \end{aligned} \right.
\]
Denote
\[
    C_q \br{r, T} = \br{ \mp i } \int\limits_0^{\pi} d \theta \int \limits_0^{2 \pi} d \phi  f \br{r, \theta, \phi} \overline{  g \br{r, \theta, \phi} } \rho (r^q )  \sin \br{\theta} r^{2-q} \big( e^{\pm i T r^q} - 1 \big) ~.
\]
Applying stationary phase method, one gets
\begin{eqnarray*}
\int_0^{+\infty} \om \big( a^{\dagger} ( e^{ \pm i t h_q} f ) a ( g )   \big) dt &=&  \lim_{T \to \infty} \int\limits_0^{+\infty} C_q \br{r, T} d r  \\ 
 &=&  \br{\pm  i} \int\limits_0^{+\infty} d r \int\limits_0^{\pi} d \theta \int \limits_0^{2 \pi} d \phi  f \br{r, \theta, \phi} \overline{  {g} \br{r, \theta, \phi} } \rho \br{ r^q } \sin \br{\theta} r^{2-q}  \\ 
 &=&  \pm i \int\limits_{\RR^3} \frac{ {f} \br{ k} \overline{  {g} \br{ k} } \rho \br{\abs{ k}^q}}{ \abs{ k}^q } dk ~.
\end{eqnarray*}

Similarly, using the anti-commutation relation for Fermi creation an annihilation operators
\[
    \int\limits_0^{\infty} \om \big( a (g)  a^{\dagger} ( e^{ \pm i t h_q} f )   \big) dt = \pm i \int\limits_{\RR^3} \frac{ {f} \br{ k} \overline{ g (k) }  \big( 1- \rho ( |k|^q ) \big) }{ |k|^q } dk ~.
\]
For $h = h_q$, $q= 1, 2$ and a Fermi reservoir define
\begin{equation} \label{eq:ConstantsFormFactors}
    \alpha_{ij} := - \sum_{p=1}^2 \int\limits_{\RR^3} \frac{   f_{ip} (k) \overline{ g_{jp} (k) } \rho ( |k|^q ) +  f_{jp} (k) \overline{  g_{ip} (k) } \big( \rho ( |k|^q )  - 1 \big)  }{ |k|^q } dk ~.
\end{equation}
For a Bose reservoir define
\begin{equation} \label{eq:ConstantsFormFactorsBoze}
    \alpha_{ij} := \sum_{p=1}^2 \int\limits_{\RR^3} \frac{ -f_{ip} (k) \overline{  g_{jp} (k)  } \rho ( |k|^q  ) +  f_{jp} (k) \overline{ g_{ip} (k)  } ( 1 + \rho ( |k|^q ) )   }{ |k|^q } dk ~.
\end{equation}

For both cases we have
\begin{equation} \label{eq:FugacityCoeffs}
\left\{ \begin{aligned} 
    & a_{ij} = i \alpha_{ij} ~, \\
    & b_{ij} = - a_{ji} ~.
\end{aligned}    \right.
\end{equation}

So, we obtain that for all permutations $\pi \in S' \br{2k}$ only one among $n!$ summands in \eqref{eq:DysonSeriesReservoirIntegral} is nonzero in the limit. In terms of Friedrichs diagrams, this summand corresponds to the diagram drawn by connecting consecutive times $\br{t_1, t_2}$, $\br{t_3, t_4}$ and soon on. This limit is exactly computed and expressed as a product of quadratic forms of the formfactors.

\section{Reduced System  Dynamic in the Weak Coupling Limit}\label{Sec4:Reduced}
In this section we derive the system reduced dynamic \eqref{eq:ReducedDynamic} in the WCL. The following series represents the reduced system dynamics and is induced by the Dyson series \eqref{eq:DysonSeriesReducedDynamic}
\begin{equation} \label{eq:TempSeries}
    \sum_{k=1}^{\infty} i^k  \sum_{ 1_k \leq \gamma \leq \nu 1_{k} } \sum_{ \pi \in S' \br{k} } Q_{\gamma, \pi} \left[ X \right] \lam^k \int\limits_{\D_k \br{\lam^{-2} t} } dt_1 \dots dt_k  \om \br{  F_{\gamma, \pi} \br{t_1, \dots, t_k} } ~.
\end{equation}
\begin{lemma} \label{le:UniformlyConver}
    Assume that Conditions \ref{cond:Formfactors}, \ref{cond:Interaction} and \ref{cond:InitState} hold. Consider operator $X = \br{\cdot, \phi} \psi$ with $\phi$, $\psi$ belonging to the set $\DD$ used in Condition \ref{cond:Interaction}. Then there exists $t_0 > 0$ such that for all $t \in \left[0, t_0 \right)$ series \eqref{eq:TempSeries} absolutely and uniformly converges for $\lam \in \br{0, 1}$.
\end{lemma}
\begin{proof}
Under Condition \ref{cond:InitState} for all $\pi \in S \br{k} $ and $1_k \leq \gamma \leq \nu 1_k$ one has that
\[
\norm{ Q_{\gamma, \pi } \left[ X \right] } \leq 2^k \br{ \max \left\{ c_{\phi}, c_{\psi} \right\}    }^k \norm{\phi} \norm{\psi} ~,
\]
where $c_{\phi}, c_{\psi}$ are constants from Condition \ref{cond:Interaction}.

If $k$ is odd, then by Theorem \ref{th:FD-RecurRelation} $k$-th term is zero. Denote $2n$-th term in \eqref{eq:TempSeries} as
\[
    s_n \br{\lam}  =  (-1)^n  \sum_{ 1_{2n} \leq \gamma \leq \nu 1_{2n}  } \sum_{ \pi \in S' \br{2n} } Q_{\gamma, \pi}  \left[ X \right] \int\limits_{\D_{2n} \br{\lam^{-2} t} } dt_1 \dots dt_{2n} \om \br{ \lam^{2n} F_{\gamma, \pi }  \br{t_1, \dots, t_{2n}}} ~.
\]
Each $F_{\gamma, \pi }  \br{t_1, \dots, t_{2n}} $ consists of $4^{2n}$ monomials and for each monomial there exist $n!$   partitioning from \eqref{eq:quasi-free}. Using Lemma \ref{le:FD-Estimate} we obtain
\[
    \Big| \lam^{ 2n } \int\limits_{\D_{2n} \br{\lam^{-2} t} } dt_1 \dots dt_{2n} \om \br{ F_{\gamma, \pi }  \br{t_1, \dots, t_{2n}}}  \Big| \leq 4^{2n} n! \frac{ \br{16C t}^n   }{n!} = \br{128 C t }^n ~.
\]
By definition, $s_n \br{\lam} $ consists on $\nu^{2n} \abs{S' \br{2n} } $ terms. We get
\[
    \abs{ s_n \br{\lam}   } \leq  \br{ 2^{11} \nu^2 \br{\max \left\{ c_{\phi}, c_{\psi}  \right\}  }^2 C t   }^n \frac{ \norm{\phi} \norm{\psi} }{2} ~.
\]
The uniform and absolute convergence of the series \eqref{eq:TempSeries} for sufficiently small $t$ follows from the Weierstrass M-test.
\end{proof}

\begin{theorem} \label{th:ReducedDynamic}
Let Conditions \ref{cond:PermutForHermite}, \ref{cond:Formfactors}, \ref{cond:Interaction} and \ref{cond:InitState} hold.
Consider the operator
\[
    H' = \sum_{i=1}^{\nu} \sum_{j=1}^{\nu} \alpha_{ij} Q_i Q_j ~,
\]
defined on $D \br{H'} = \DD$, where the coefficients $\alpha_{ij}$ are defined by \eqref{eq:ConstantsFormFactors} for Fermi case and by \eqref{eq:ConstantsFormFactorsBoze} for Bose case. Assume that the operator $H^S - \lam^2 H'$ with the domain of definition $\DD$ is essentially self-adjoint for all $\lam \in \br{0, 1}$. Denote closure operator
\[
    \wt{H}_{\lam} = \overline{ \lam^{-2} H^S - H'  } ~.
\]
Consider $X \in \UU^S$. Then there exists $t_0 > 0 $ such that for all $t \in \left[0,  t_0 \right) $ the reduced dynamic for the system \eqref{eq:ReducedDynamic} satisfies
\[
    \wt{\Tau}_{\lam} \br{t} X  = e^{ i  \wt{H}_{\lam}  t   }  X e^{ -i \wt{H}_{\lam}  t   }  + O_{\norm{\cdot} } \br{\lam},\quad \lam \downarrow 0.
\]
\end{theorem}
\begin{proof}
First, we claim that the operator $H'$ is a Hermitian (i.e., symmetric) operator. Indeed, due to Condition \ref{cond:PermutForHermite} one has
\[
\br{ \alpha_{ij} Q_i Q_j }^{\dagger} = \overline{ \alpha_{ij} } Q_j^{\dagger} Q_{i}^{\dagger} = \alpha_{ \sigma \br{j} \sigma \br{i} } Q_{\sigma \br{j}} Q_{\sigma \br{i}} ~.
\]
    
Assume that $t$ is small enough so that  series \eqref{eq:TempSeries} uniformly converges for $\lam \in \br{0, 1}$, and take the term  by term limit. Using Theorem \ref{th:FD-RecurRelation}, we obtain
\[
\wt{\Tau}_{\lam} \br{t} X = \Tau^S \br{ \lam^{-2} t} \br{X + \sum_{n=1}^{\infty}  \br{-1}^n \br{ \sum_{1_{2n} \leq \gamma \leq \nu 1_{2n} } \sum_{ \pi \in S' \br{2n}  }  U \br{\gamma, \pi} Q_{\gamma, \pi} \left[ X \right] } + R_n \br{\lam} \left[ X \right]     } ~,
\]
where $R_n \br{\lam} $ can be estimated as
\[
    \norm{ R_n \br{\lam} \left[X \right]  } \leq  \nu^{2n} 2^{2n-1}  \br{\max \left\{ c_{\phi}, c_{\psi} \right\}}^{2n} \frac{n A \br{C't}^n }{\sqrt{t}} \lam ~,
\]
where $A, C'$ are some constants. For sufficiently small $t$ we have
\[
    \sum_{n=1}^{\infty} \norm{ R_n \br{\lam} \left[ X\right]  } =  O \br{\lam} ~.
\]

Consider $n$-th term in \eqref{eq:TempSeries}
\[
    A_n \left[X \right] = \br{-1}^n \sum_{1_{2n} \leq \gamma \leq \nu 1_{2n} } \sum_{ \pi \in S' \br{2n}  }  U \br{\gamma, \pi} Q_{\gamma, \pi} \left[ X \right] ~.  
\]
Using recurrence relation  from Theorem \ref{th:FD-RecurRelation}, we get
\begin{eqnarray*}
    A_{n} \left[X \right] = \br{-1} \frac{t}{n}  \sum_{j=1}^{\nu} \sum_{j=1}^{\nu} && a_{ij} Q_i Q_j A_{n-1} \left[X \right] - a_{ij} Q_i A_{n-1} \left[X \right] Q_j  \\
    && - b_{ij} Q_j A_{n-1}  \left[X \right] Q_i +  b_{ij} A_{n-1} \left[X \right] Q_j Q_i ~.
\end{eqnarray*}
Taking into account relations \eqref{eq:FugacityCoeffs}, we obtain
\[
    A_n \left[ X \right] = -i \frac{t}{n} \comm{ H' }{ A_{n-1} \left[ X \right]  } ~.
\]
Let $\mathcal{L} \left[X \right] = i \comm{ - H' }{X} $. Then we get
\begin{equation} \label{eq:DysonSeriesFinal}
    \wt{\Tau}_{\lam} \br{t} X = \sum_{n=0}^{\infty} \frac{t^n}{n!} \Tau^S \br{\lam^{-2} t} \br{ \mathcal{L}^n \left[ X \right]  } ~.    
\end{equation}

Now, consider the Cauchy problem for the Heisenberg equation
\begin{equation} \label{eq:InitialProblemRedDyn}
\left\{ \begin{aligned}
    & \frac{ d X \br{t} }{d t} = i \big[ H^S-\lam^2 H'  , X \big] ~, \\
    & X \br{0} = X_0 ~,
\end{aligned} \right.
\end{equation}
where $X_0 \in \UU^S$.  For solving \eqref{eq:InitialProblemRedDyn} consider the following Cauchy problem
\begin{equation} \label{eq:InitialProblemRedDynReplacement}
\left\{ \begin{aligned}
    & \frac{ d C \br{t} }{d t} = i \lam^2 \big[ -  e^{- i  H^S t }  H' e^{i  H^S t}   , C \br{t} \big] ~, \\
    & C \br{0} = X_0 ~.
\end{aligned} \right.
\end{equation}
Taking into account Conditions \ref{cond:PermutForHermite} and \ref{cond:Interaction}, we obtain the solution of \eqref{eq:InitialProblemRedDynReplacement} as
\[
    C \br{t} = X_0 + \sum_{n=1}^{\infty} \frac{t^n}{n!} \br{\lam^2 \mathcal{L}}^n \left[X_0 \right] ~.
\]
The proof of converge of this series is similar to the proof of Lemma \ref{le:DS-converge}. Finally we obtain that the solution of \eqref{eq:InitialProblemRedDyn} is represented by series \eqref{eq:DysonSeriesFinal} as
\begin{equation} \label{eq:DysonTempReduced}
    X \br{t} = \Tau^S \br{t} C \br{t} = \sum_{n=0}^{\infty} \frac{t^n}{n!}  \Tau^S \br{t} \br{\lam^2 \mathcal{L}}^n \left[X_0 \right] ~.  
\end{equation}
The solution can be defined in the weak sense, as described in section \ref{sec:DS-definition}. We can see that series \eqref{eq:DysonSeriesFinal} and  \eqref{eq:DysonTempReduced} are match if we make time scaling $t \to \lam^{-2} t$ in the second series.  In other words, equation \eqref{eq:InitialProblemRedDyn} is a master equation for the reduced system dynamic in the WCL for the considered model.
\end{proof}
\begin{remark}
We have derived the reduced system dynamics only for some subset of the system observable algebra $\UU^S$. The formula for the evolution of any pure state $\psi$ with $\psi \in \DD$ in the weak coupling limit follows from the obtained result as
$$
\psi \mapsto   e^{ -i  \wt{H}_{\lam} t } \psi  ~.
$$
\end{remark}

\begin{example}
    Let us show that for the interaction Hamiltonian in Example \ref{example:CommutativeCase} the operator $\wt{H}_{\lam}$ with domain $F$ is essentially self-adjoint. This operator is a Hermitian (i.e., symmetric) operator and in momentum representation one has
\[
     (H^S - \lam^2 H') f  = \om \br{ k}  {f} \br{ k} ~,
\]
where $\om \br{ k}$ is a real-value  function of $ k \in \RR^3$. Since $\om \br{ k}$ is a real-valued function, if $ \br{ \om \br{ k} \pm i }  {f} \br{ k} = 0 $ for some $ {f} \br{ k} $ then $ {f} \br{ k} = 0$ in $L_2 \br{\RR^3}$. We obtain 
\[
    \ker  \bigl((H^S - \lam^2 H')^{\dagger} \pm i  \bigr) = \left\{ 0 \right\} ~.
\]
Then the operator $ H^S - \lam^2 H'$ defined on the domain $\DD$ is essentially self-adoint.
\end{example}

\section{Numerical comparison of the free and limiting dynamics}\label{Sec6:Simulations}
In this section, we analyze and compare the dynamics induced by the free system  and the modified limiting Hamiltonians $H^S, \widetilde{H}_{\lam}$ respectively by considering simulation of the dynamics of a Gaussian wave packet.

Explicitly, we consider as the initial system state the Gaussian wave packet
\[
    \phi \br{x}  = \pi^{ - \frac{3}{4} } \exp{ - \frac{ \abs{x}^2 }{2} } ~,
\]
the free and one-particle  Hamiltonians $h = H^S = h_2$, and the interaction Hamiltonian 
\begin{equation} \label{eq:InteractionExamples}
    V = \br{ c_1 P_1 + c_2 P_2 + c_3 P_3 } \otimes (a^{\dagger} \br{f} + a \br{f}) ~, 
\end{equation}
where $P_j = -i \partial_j$ is the vector momentum operator and $c_j \in \RR$ are some constants. Let $\phi_t (x) = e^{- i H^S t} \phi (x)$, $\wt{\phi}_t (x) = e^{ -i ( H^S -\lam^2 H') t} \phi \br{x}$ be the free and perturbed in WCL evolution, where $H'$ is the operator from Theorem \ref{th:ReducedDynamic} for interaction \eqref{eq:InteractionExamples} and $\lam$ is small parameter. We compare probability densities in this case, i.e. $\rho_{t} \br{x} = | \phi_t \br{x} |^2 $ and $\wt{\rho}_{t} \br{x} = | \wt{\phi}_t \br{x} |^2 $. For this, define the following time-dependent matrices
\begin{eqnarray*}
     M_0 \br{t}  & = & \br{1 + 2i t} \mathbb{I}_3 ~, \\
     M \br{t}  & = &  \br{1 + 2i t} \mathbb{I}_3 - 2 i \alpha \br{ c \cdot c^{\mathrm{T}  }}  \lam^2 t ~,
\end{eqnarray*}
where $c = \br{c_1, c_2, c_3}^{\mathrm{T}} $, $\mathbb{I}_3$ is the identity $3 \times 3$ matrix and $\alpha$ is defined by~\eqref{eq:ConstantsFormFactors} as 
\[
    \alpha = \int\limits_{\RR^3} \frac{ \abs{ {f} \br{ k} }^2 \big( 1 - 2 \rho  ( |k|^2 ) \big)   }{ |k|^2 } dk ~.
\]
We have
\begin{eqnarray*}
       \rho_{t} \br{x} &=& \pi^{-\frac{3}{2}} \abs{ \det M_0 \br{t} }^{-1} \abs{ \exp{ -\frac{1}{2} x^{\mathrm{T}} M_0^{-1} \br{t} x }  }^2 ~, \\
       \wt{\rho }_{t} \br{x} &=& \pi^{-\frac{3}{2}} \abs{ \det M\br{t} }^{-1} \abs{ \exp{ -\frac{1}{2} x^{\mathrm{T}} M^{-1} \br{t} x }  }^2 ~.
\end{eqnarray*}
Using Taylor expansion around $\lam=0$ we get
\begin{equation} \label{eq:InverseMatrixTayExpans}
    M^{-1} \br{t} \sim \br{1 + 2it}^{-1} \mathbb{I}_3  + 2i \alpha t  \lam^2 \br{1 + 2it}^{-2} \br{c \cdot c^{\mathrm{T}}} ~.
\end{equation}
With \eqref{eq:InverseMatrixTayExpans} we obtain
\[
    \wt{\rho }_{t} \br{x} \sim C \br{t} e^{ -G_t \br{x} } ~,
\]
where
\[
    G_t \br{x} = \frac{ \abs{x}^2 }{ 4 t^2 + 1 } + \lam^2 \frac{16 \alpha t^2}{\br{ 4 t^2+1}^2 } \br{ c^{\mathrm{T} }  x  }^2
\]
and $C \br{t}$ is a normalization factor. So, we have that $\rho_t \br{x}$ is a spherically symmetric Gaussian  with increasing dispersion and $\wt{ \rho }_t \br{x}$ is a deformed Gaussian that decreases most slowly (resp., faster) in the direction orthogonal to vector $c$ if $\alpha$ is positive (resp., negative).

For simulations set $\alpha = 4.0$ and $\lam = 0.1$.   

Fig.~\ref{fig:snapshot} shows snapshots of the wave packet dynamics under the free system Hamiltonian (upper row) and under the modified Hamiltonian (middle and bottom rows, cases $c_1= 4.0$, $c_2 = -4.0$, $c_3 = 0.0$ and $c_1 = 4.0$, $c_2 = 0.0$, $c_3 = 0.0$, respectively) at four time instants, starting time is $0.0$ and the final time is $1.5$. We chose these values of $\alpha$ and $c_j$  for simplicity of the animation, so that the differences between the free and perturbed dynamics are clearly visible.  Fig.~\ref{fig:movie} (Multimedia available online) shows the movie with continuous time free evolution and perturbed evolution for case $c_1 = 4.0, c_2 = -4.0, c_3 =0.0$ of the wave-packet. Since our results hold for dimension $d > 2$ (see Sec.~\ref{Sec5:Discussion}), we plot the probability density functions $\rho_t \br{x}, \wt{ \rho_t } \br{x}$ in projection onto the plane $x_3=0$, i.e. we plot $\rho_t(x_1,x_2,0)$ and $\wt{\rho}_t(x_1,x_2,0)$.

\begin{figure}
\centering
\includegraphics[width = \linewidth]{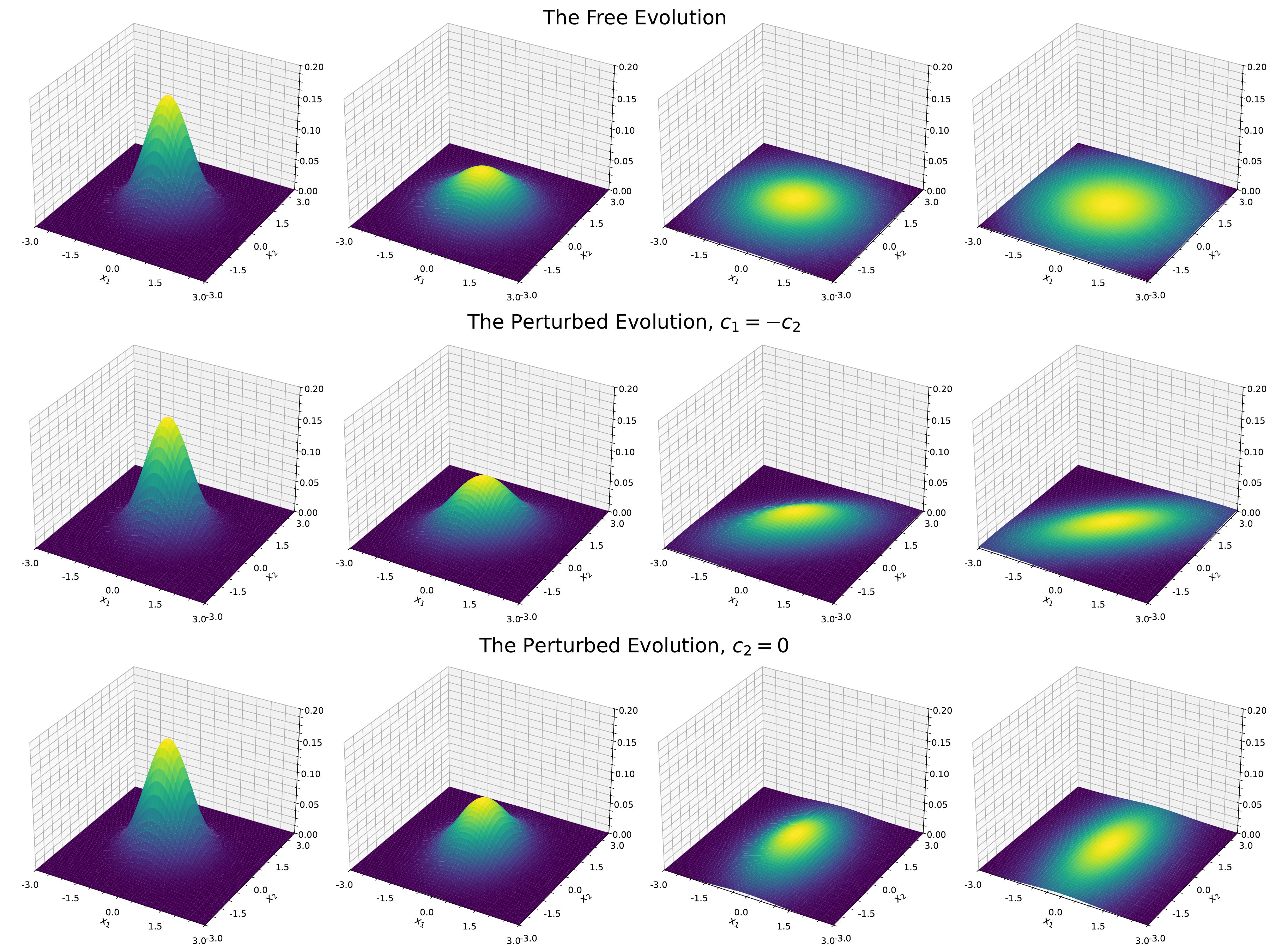}  
\caption{\label{fig:snapshot} Evolution of the Gaussian wavepacket under the free system Hamiltonian $H^{\rm S}$ (upper row) and under the modified limiting Hamiltonian with $ c_1 = - c_2 $ (middle row) and $c_2 = 0$ (bottom row) at four different time instants [0.0, 0.5, 1.0, 1.5].}
\end{figure}

\begin{figure}
\centering
\includegraphics[width = \linewidth]{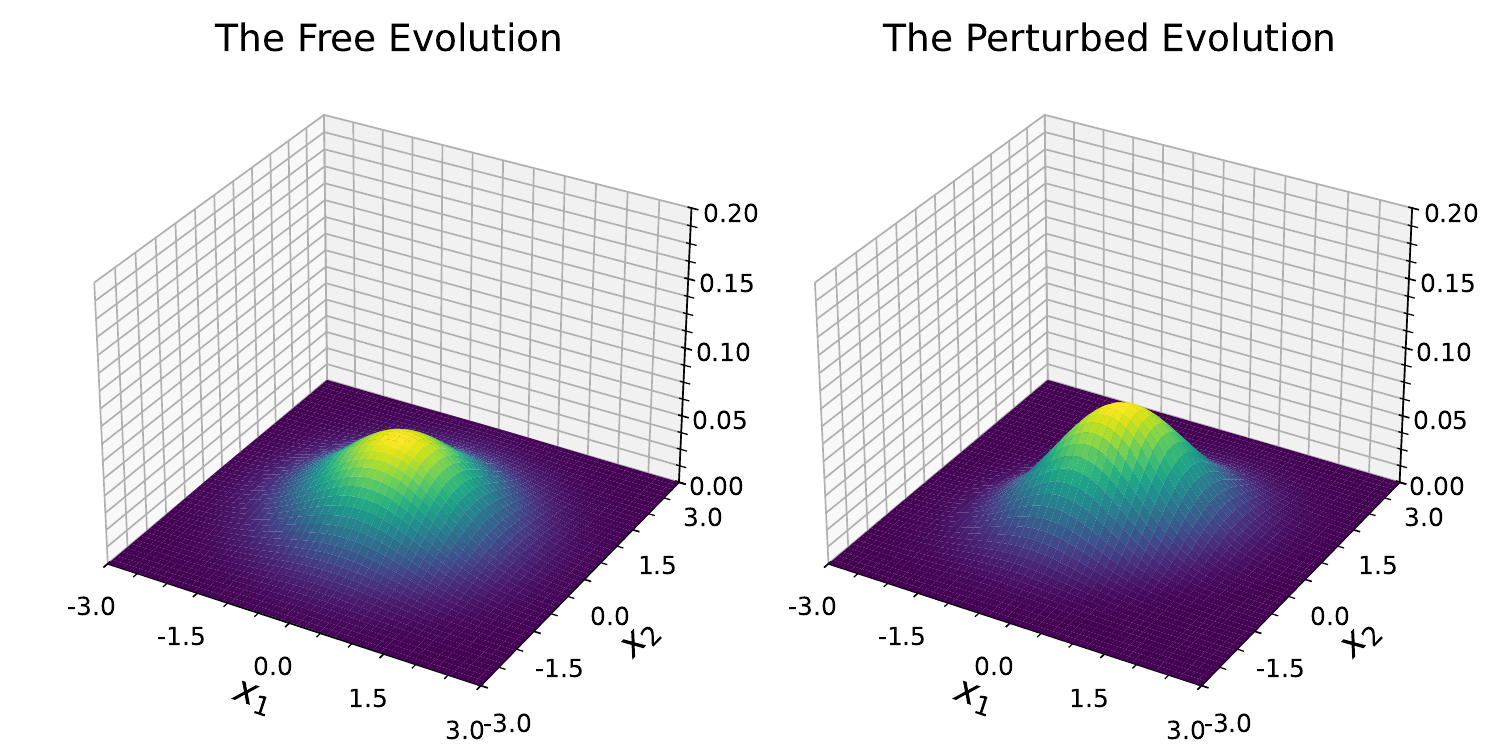}  
\caption{Evolution of the Gaussian wavepacket under free Hamiltonian $H^{\rm S}$ (left) and under modified lmiting Hamiltonian $\widetilde{H}_{\lam}$ (right) (Multimedia available online in 
the ancillary file EvolutionOfGaussianWavepacket.mp4).
\label{fig:movie}}
\end{figure}

\section{Discussion}\label{Sec5:Discussion}
In this paper we rigorously derive the reduced dynamic of a quantum system interacting with a reservoir consisting of Fermi or Bose particles. We consider an infinite-dimensional quantum system with unbounded system operators with continuous spectrum. Motivated by the model of an atom interacting with electro-magnetic field, we impose a weak commutativity condition for the free system Hamiltonian and the system part of the interaction. For such system, we  describe all nonzero in the limit diagrams and correlation functions of the reservoir, and compute their limiting values. Then we derive the limiting reduced system dynamics which has the form of a unitary evolution with some modified Hamiltonian. The modified Hamiltonian has the form of the free Hamiltonian and some additional term which can be interpreted as Lamb shift. Moreover, we have found upper estimates on the rest terms that allowed to estimate the convergence rate of the reduced system dynamics to the limiting dynamics; we prove that the sum of vanishing in the limit terms behaves as $O \br{\lam} $. 

The problem of proving self-adjointness (or essential self-adjointness) of Hamiltonians of the form  $H_0 + \lam V$ is a difficult  problem\cite{Arai::HamiltonianQuantumElectrodynamics, Jaksik::FermiGoldenRule}. In our analysis, we do not rely on self-adjointness of the full Hamiltonian $H_0 + \lam V$ or, for example, on that the domain $\DD$ is a domain of the essential self-adjointness for this operator. We request essential  self-adjointness for the operator $H^S - \lam^2 H'$ in Theorem \ref{th:ReducedDynamic} so that exponent of this operator is correctly defined when $H^S$ and $Q_j$ are unbounded operators. In other hand, Conditions \ref{cond:Formfactors} -- \ref{cond:InitState} are sufficient to define Dyson series (see Lemma \ref{le:DS-converge} ) and obtain the limiting master-equation \eqref{eq:InitialProblemRedDyn}. 

The absence of the dissipative part in the limiting dynamics can be understood from the general WCL theory. In the WCL, the reduced system dynamics is determined by quantities involving the limit as $\lambda\to 0$ of integration over $t\in\mathbb R_+$ of correlation functions with kernel $\frac{1}{\lambda^2}e^{i\omega(k)t/\lambda^2}$. In the sense of distributions, according to Sokhotski–Plemelj theorem
\[
\int\limits_0^\infty e^{\pm i\omega t}dt=\pi\delta(\omega)\pm i P.V.\frac{1}{\omega} ~.
\]
The first and the second terms in the r.h.s. of this equality determine contributions of the interaction with the reservoir to the dissipative and Hamiltonian parts of the limiting master equation for the reduced system dynamics. In our case $\omega$ is the function on $\mathbb R^3$ and $\delta(\omega)$ is a surface delta-function concentrated on the surface $\omega(k)=0$ in $\mathbb R^3$ (surface delta-functions are discussed for example in Sec.~1.7 of Ref.~\onlinecite{vladimirovbook}). In the considered cases of either $\omega(k)=|k|$ or $\omega(k)=|k|^2$, the surface $\omega(k)=0$ reduces to only one point and delta-function on this surface becomes zero (it is not the same as point delta-function centered at the origin). For example, delta-function $\delta_{S_r}$ of the sphere of radius $r$ acts on a test function $f\in\SS(\mathbb R^3)$ as $(\delta_{S_r},f)=\int\limits_S dS f(x)$ and as $r\to 0$ such integrals for any test function $f$ converge to zero. In other examples, such as for interaction of a two-level atom with a reservoir, $\omega(k)$ has the form $\omega(k)=|k|-\omega_0$ or $\omega(k)=|k|^2-\omega_0$, where for two-level atom $\omega_0$ is the non-zero transition frequency between the two levels. In this case, surface $\omega(k)=0$ is not reduced to a point and determines a non-trivial surface delta-function. In this sense, under the GRWA the single Bohr frequency $\omega_0$ contributes to the limiting dynamics. As mentioned above, absence of the dissipative part in the reduced dynamics is sometimes called as Quantum Cheshire Cat effect~\cite{ALVBook}.

Some of our estimates can be improved, such as for example estimate of the absolute value  of the Jacobian in the proof of Lemma~\ref{le:ZeroFD}. However, in the analysis of this work it can only lead to an increase of the constant $t_0$ in main Theorem \ref{th:ReducedDynamic}, but because of its implicit definition, this increase is not significant.

The results of this work also hold for the space $\RR^d$ with any dimension $d \geq 3$. For $d \leq 2$, the two-point correlation function $ c(\Delta t) = \om \big( a^{\dagger} (e^{i \Delta t h} f) a (  g)  \big) $  can be non-integrable on $(- \infty, + \infty)$ .

A more general condition for Lemma~\ref{le:fugacityIneq} form-factors $f_i, g_j$ and density $\rho (h)$ is that $\rho(k) f_i(k) g_j (k)$ belongs to $\mathfrak G$.

If we consider vector $Q = \br{Q_1, Q_2, \dots, Q_{\nu}}^{\mathrm{T}}$ where $Q_j$ is operator from \eqref{eq:Interaction}, and matrix $A_{ij} = \alpha_{ij}$, where $\alpha_{ij}$ are defined by \eqref{eq:ConstantsFormFactors}, \eqref{eq:ConstantsFormFactorsBoze} for Fermi and Bose particles, respectively, then the additional Hamiltonian $H'$ appearing in Theorem \ref{th:ReducedDynamic} can be expressed as
\[
    H' = Q^{ \mathrm{T} } A Q ~.
\]
In general matrix $A$ is not Hermitian. But if each operator in \eqref{eq:Interaction} is Hermitian, then $f_i = g_i$ for all $i = 1, \dots, \nu$ and $A$ is a Hermitian matrix.  In general $H'$ can be represented as a sum of Hermitian operators due to Condition \ref{cond:PermutForHermite}. If $i = j$ and $\sigma \br{i} = i $ then $\alpha_{ij} Q_i Q_j $ is a Hermitian operator, otherwise $\alpha_{ \sigma \br{j} \sigma \br{i} } Q_{\sigma \br{j} } Q_{\sigma \br{i}} + \alpha_{ij} Q_i Q_j$ is a Hermitian operator.

\begin{acknowledgments}
The authors are deeply grateful to Alexander E. Teretenkov for useful comments. This work was supported by Ministry of Science and Higher Education of the Russian Federation (Grant No. 075-15-2024-529).
\end{acknowledgments}

\section*{AUTHOR DECLARATIONS}
\noindent {\bf Conflict of Interest}

The author has no conflicts of interest to disclose.

\section*{DATA AVAILABILITY}

Data sharing is not applicable to this article as no new data were created or analyzed in this study.

\appendix

\section{Proof of Lemma \ref{le:EstimateIntegral}} \label{app:ProofLemma}
\begin{proof}
For the first inequality 
\begin{eqnarray*}
 \lam^{2p} \int\limits_{0}^{\lam^{-2}t} \tau^p \chi \br{\tau} d \tau &\leq&   \lam^{2p} \br{ \int\limits_0^1 \tau^p d \tau + \int\limits_{1}^{\lam^{-2}t} \tau^{p - \frac{3}{2} } d \tau  } \\ 
 &\leq&  \frac{ \lam t^{p -\frac{1}{2}} }{p - \frac{1}{2}} - \lam^{2p} \br{ \frac{1}{p-\frac{1}{2}} - \frac{1}{p+1} } \leq 2 t^{p - \frac{1}{2}} \lam.
\end{eqnarray*}
We prove the second inequality by induction on $k$. The base of induction is
\[
    G_{1, p} \br{\lam} = \lam^{2p} \int\limits_0^{\lam^{-2}t} \tau^p \chi \br{\tau} d \tau \leq 2  t^{p-\frac{1}{2}} \lam,
\]
holds for all $p \geq 1$. 

The application of the binomial expansion for $k > 1$ yields
\begin{eqnarray*}
    G_{k, p} \br{\lam} & = & \sum_{q=1}^p \binom{p}{q} \lam^{2p} \int\limits_{S_k \br{\lam^{-2}t}} dt_1 \dots dt_k t_1^{q} \br{t_2 + \dots + t_k}^{p-q} \chi \br{t_1} \dots \chi \br{t_k}  \\
    &&+ \lam^{2p} \int\limits_{S_k \br{\lam^{-2}t}} dt_1 \dots dt_k \br{t_2 + \dots + t_k}^p \chi \br{t_1} \dots \chi \br{t_k}.
\end{eqnarray*}
For $q \geq 1$ we  obtain
\begin{eqnarray*}
 && \lam^{2p} \int\limits_{S_k \br{\lam^{-2}t}} dt_1 \dots dt_k t_1^{q} \br{t_2 + \dots + t_k}^{p-q} \chi \br{t_1} \dots \chi \br{t_k}  \\
 && \leq  \lam^{2p} \int\limits_{S_k \br{\lam^{-2}t}} dt_1 \dots dt_k \br{\lam^{-2}t}^{p-q} t_1^q \chi \br{t_1} \chi \br{t_2} \dots \chi \br{t_k}  \\ 
 && \leq t^{p-q} \lam^{2q} \int\limits_{ \br{0, \lam^{-2}t}^{k} } dt_1 \dots dt_k t_1^q \chi \br{t_1} \chi \br{t_2} \dots \chi \br{t_k}  \\ 
 &&   \leq t^{p-q} \left( \int\limits_0^{+\infty} \chi \br{\tau} d \tau \right)^{k-1} \br{\lam^{2q} \int\limits_0^{\lam^{-2}t} \tau^q \chi \br{\tau} d \tau  } \leq 2^k \lam t^{p-\frac{1}{2}}.
\end{eqnarray*}
For $q=0$ we get
\begin{eqnarray*}
&& \lam^{2p} \int\limits_{S_k \br{\lam^{-2}t}} d t_1 \dots d t_k \br{t_2 + \dots + t_k}^p \chi \br{t_1} \dots \chi \br{t_k}  \\ 
&& = \lam^{2p} \int\limits_0^{\lam^{-2} t} d t_1 \chi \br{t_1} \int\limits_{ S_{k-1} \br{\lam^{-2} t - t_1}  } d t_2 \dots d t_{k} \br{ t_2 + \dots + t_k  }^p   \chi \br{t_2} \dots \chi \br{t_k} \\ 
&& \leq \br{ \int\limits_0^{+\infty} \chi \br{\tau} d \tau  } G_{k-1, p} \br{\lam} = 2 G_{k-1, p} \br{\lam}.
\end{eqnarray*}
Finally, by induction assumption
\[
    G_{k, p} \br{\lam} \leq  2^k \lam t^{p-\frac{1}{2}} \sum_{q=1}^p \binom{p}{q} + 2 G_{k-1 ,p} \br{\lam} \leq 2^{k+p} \lam t^{p - \frac{1}{2}} + 2 \br{k-1} 2^{k-1+p} t^{p -\frac{1}{2}} \lam \leq k 2^{k+p} t^{p -\frac{1}{2}} \lam.
\]
\end{proof}

\end{document}